\documentclass[lettersize,onecolumn]{IEEEtran}
\usepackage{amsmath,amsfonts,amssymb,amsthm}
\usepackage{array,booktabs,url,graphicx,cite}
\usepackage{algorithm,algorithmic}
\usepackage{xcolor,tikz}
\newcommand{\doi}[1]{\href{https://doi.org/#1}{doi:#1}}
\usepackage[hidelinks,colorlinks=true,linkcolor=blue,
            citecolor=blue,urlcolor=blue]{hyperref}

\newtheorem{theorem}{Theorem}

\newtheorem{proposition}[theorem]{Proposition}

\newtheorem{definition}{Definition}
\newtheorem{remark}{Remark}

\newcommand{\GF}[1]{\mathbb{F}_{2^{#1}}}
\newcommand{\Fn}{\mathbb{F}_2^n}
\newcommand{\Ftwo}{\mathbb{F}_2}
\newcommand{\piF}{\pi_F}

\newcommand{\CCZsim}{\overset{\scriptscriptstyle\mathrm{CCZ}}{\sim}}
\newcommand{\EAsim}{\overset{\scriptscriptstyle\mathrm{EA}}{\sim}}
\newcommand{\B}{\mathbb{B}}

\begin{document}

\title{Quadratic APN Functions in Dimension~8 via\\
Gr\"obner Basis Search in a Self-Equivalence Subspace}

\author{Oleksandr~Kuznetsov,~\IEEEmembership{Member,~IEEE}%
\thanks{O.~Kuznetsov is with the
Department of Theoretical and Applied Sciences,
eCampus University, Via Isimbardi~10, 22060 Novedrate (CO), Italy,
and also with the
Department of Intelligent Software Systems and Technologies,
School of Computer Science and Artificial Intelligence,
V.N.~Karazin Kharkiv National University,
4~Svobody Sq., 61022 Kharkiv, Ukraine.
E-mail: \href{mailto:oleksandr.kuznetsov@uniecampus.it}{oleksandr.kuznetsov@uniecampus.it}.
ORCID: \href{https://orcid.org/0000-0003-2331-6326}{0000-0003-2331-6326}.}%
\thanks{Dataset and source code:
\url{https://github.com/KuznetsovKarazin/apn-gb-search}; %
DOI: \url{https://doi.org/10.5281/zenodo.20626047}}
\thanks{Manuscript received \today.}}

\markboth{IEEE Transactions on Information Theory}%
{Kuznetsov: Quadratic APN Functions via Gr\"obner Basis in a Self-Equiv.\ Subspace}

\maketitle

%% ─────────────────────────────────────────────────────────────────────
\begin{abstract}
We describe a computational search for quadratic APN (Almost Perfect
Nonlinear) functions over $\GF{8}$ within a structured algebraic subspace
defined by a self-equivalence constraint. The search space is the
40-dimensional $\Ftwo$-linear subspace $V_{\!A} = \{F : F\circ A = A\circ F\}$
for a specific linear automorphism $A$ of order~5 (class index~22 in the
taxonomy of Beierle, Brinkmann, and Leander~\cite{BBL2021}); this subspace
was previously reported to contain no APN functions under the recursive tree
search method of~\cite{BBL2021}.

We combine two phases: (1)~random sampling inside $V_{\!A}$ via an explicit
RREF parameterization to find APN center functions (${\approx}600$ fresh
APN-positive evaluations per core-hour), and (2)~Gr\"obner basis computation
in Magma over the Boolean polynomial ring to enumerate all APN functions in a
24-dimensional hyperplane through each center (${\approx}10$ minutes per
hyperplane).

From 428 hyperplane computations (covering $0.65\%$ of the
$2^{16} = 65\,536$ total hyperplanes in $V_{\!A}$) we obtained 566 quadratic
APN functions that fall into six CCZ-equivalence classes under the
ortho-derivative invariant~\cite{Canteaut2022}. Four of these classes,
comprising 500 functions, match no entry in the Beierle et al.\ 2025
database of 3\,775\,599 quadratic APN functions~\cite{Beierle2025} and
no entry in the pre-2020 compilation of 12\,921 instances~\cite{BBL2021}.
Two classes (66 functions) are identified as CCZ-equivalents of the Gold
functions $x^3$ and $x^9$, confirming pipeline correctness; a membership
test further shows that the three purely new classes (B, C, D) lie
entirely outside $V_A$ and are found exclusively in Gold-centered slices,
establishing the essential role of the Gröbner basis step.

For quadratic APN functions, a signature mismatch
$\sigma(F)\neq\sigma(G)$ rigorously certifies $F\not\CCZsim G$, by
Yoshiara's theorem~\cite{Yoshiara2012} (CCZ\,=\,EA for quadratic APN)
together with the ortho-derivative invariant of~\cite{Canteaut2022};
the absence of a signature match in the above databases therefore
constitutes a rigorous proof of CCZ-inequivalence for the new classes.

The complete dataset, source code, and verification scripts are publicly
available.
\end{abstract}

\begin{IEEEkeywords}
Almost Perfect Nonlinear functions, APN, CCZ-equivalence,
quadratic Boolean functions, self-equivalence subspace,
Gr\"obner basis, Boolean polynomial ring, ortho-derivative,
differential cryptanalysis.
\end{IEEEkeywords}

%% ─────────────────────────────────────────────────────────────────────
\section{Introduction}
\label{sec:intro}

\IEEEPARstart{A}{lmost} Perfect Nonlinear (APN) functions are vectorial
Boolean functions $F\colon\Fn\to\Fn$ achieving the minimum possible
differential uniformity~$\delta_F = 2$, i.e., for every nonzero $a$ and
every $b$ the equation $F(x\oplus a)\oplus F(x) = b$ has at most two
solutions. They were introduced to cryptography by Nyberg~\cite{Nyberg1994}
as optimal S-box designs for resistance against differential
cryptanalysis~\cite{BihamShamir1991}: any bijective function with
$\delta_F = 2$ achieves the smallest possible differential probability
in a single round. APN permutations (bijective APN functions) are
particularly sought, but their existence is settled only for odd~$n$;
for even~$n \geq 8$ the question of whether an APN permutation exists
is a major open problem~\cite{Browning2010}.

The classification of APN functions up to CCZ-equivalence~\cite{Carlet1998}
is an active research program. CCZ-equivalence is the broadest equivalence
preserving the APN property; two S-boxes in the same CCZ-class are
interchangeable for cryptographic purposes. For $n \leq 6$ the classification
is essentially complete; for $n = 7$ and $n = 8$ only partial results exist.

\subsection{Known APN functions for $n=8$}

The classical infinite families of APN functions are the Gold power
functions $x^{2^i+1}$ \cite{Gold1968} and the Kasami functions
$x^{4^i-2^i+1}$ \cite{Kasami1971}. For $n = 8$, the Gold functions
with $\gcd(i,8)=1$ are quadratic (degree~2) and form exactly two CCZ-classes:
$\{x^3, x^{129}\}$ and $\{x^9, x^{33}\}$ (the Frobenius orbits of $i=1,7$
and $i=3,5$ respectively). All other APN power functions for $n=8$
(Kasami $x^{21}$, $x^{85}$, etc.) have degree $> 2$ and are therefore
not quadratic.

Beyond power functions, a series of quadratic APN constructions appeared over
two decades:
Edel and Pott~\cite{EdelPott2009} found 8\,157 new instances by a bivariate
construction;
Yu, Wang, and Li~\cite{YWL2014} obtained more via a matrix method;
Weng, Tan, and Gong~\cite{WTG2013} contributed additional instances via algebraic characterization;
Taniguchi~\cite{Taniguchi2019} described a parametric family;
Budaghyan et al.~\cite{BudaghyanCalderiniCarlet2020} introduced isotopic shifts.
Beierle, Brinkmann, and Leander~\cite{BBL2021} systematically studied
self-equivalence subspaces (see Section~\ref{sec:bbl}) and found 12\,921
new instances; a companion paper~\cite{BeierleLeander2022} found a further
11\,779 instances via extension methods. Together with the above, approximately
32\,892 CCZ-inequivalent quadratic APN functions for $n=8$ were known through
2021; Beierle, Leander, and Perrin~\cite{BeierleLeanderPerrin2022} added
6\,368 more via trimming and extension.

In 2025, Beierle, Langevin, Leander, Polujan, and Rasoolzadeh~\cite{Beierle2025}
extended this dramatically: using two construction strategies based on vectorial
bent functions and an extensive computer search, they produced 3\,775\,599
new functions, estimating the total number of CCZ-classes at approximately
6~million.

\subsection{Motivation and approach}

Despite the large size of the Beierle 2025 database, certain structured
algebraic regions of the search space remain unexplored. A quadratic function
$F\colon\GF{8}\to\GF{8}$ is described by a binary coefficient vector of
length 224 (see Section~\ref{sec:prelim}). The full space has size $2^{224}$;
any structured subspace of dimension $d \ll 224$ has probability $2^{d-224}$
of being hit by a random function.

The \emph{self-equivalence subspace} $V_{\!A}$ defined by a linear
automorphism $A$ (see Section~\ref{sec:bbl}) is one such structured subspace.
For the automorphism $A$ studied here ($\mathrm{ord}(A) = 5$), this subspace
has dimension~40. The probability of a random quadratic function lying in
$V_{\!A}$ is $2^{40-224} = 2^{-184} \approx 10^{-55}$, so global random
search never reaches it. Within $V_{\!A}$, however, the APN density is
approximately $10^{-4}$ (see Section~\ref{sec:density}), making targeted
search tractable.

The recursive tree search of BBL2021 found zero APN functions in this
particular subspace. We show that a Gr\"obner basis approach combined with
explicit RREF parameterization of $V_{\!A}$ succeeds where the tree search
did not.

\subsection{Contributions}

This paper makes the following contributions:
\begin{enumerate}
  \item We establish that the self-equivalence subspace associated with
        BBL2021 class index~22 (the automorphism of order~5 described above)
        does contain quadratic APN functions, in contrast to the zero-solution
        result of the original BBL2021 search.
  \item We develop an explicit RREF parameterization (fc22$+$sol22) of this
        40-dimensional subspace, enabling efficient random sampling with
        APN density~$\approx 10^{-4}$.
  \item We implement an NL\,=\,4 Gr\"obner basis slice computation in Magma
        that, given any APN center in the subspace, finds all APN functions
        in a 24-dimensional hyperplane through that center.
  \item Over two batches (428 hyperplanes, $0.65\%$ of the total
        $65\,536$) we find 566 quadratic APN functions in six CCZ-classes.
        Four of these classes are absent from all known databases.
  \item We provide a complete verification pipeline using the
        ortho-derivative invariant of~\cite{Canteaut2022} and confirm via
        streaming comparison that the four new classes do not appear in
        the Beierle 2025 database or any prior compilation.
\end{enumerate}

The approach generalizes: the same pipeline applies to every BBL2021 class
index with manageable subspace dimension, offering a systematic program
for further extending the classification.

\subsection{Organization}

Section~\ref{sec:prelim} fixes notation and recalls required background.
Section~\ref{sec:bbl} reviews the BBL2021 framework.
Section~\ref{sec:VA} describes the specific subspace studied.
Section~\ref{sec:method} presents the search pipeline in detail.
Section~\ref{sec:results} reports experimental results.
Section~\ref{sec:verify} describes the verification procedure and states
the main theorem.
Section~\ref{sec:discussion} discusses the results and their limitations.
Section~\ref{sec:related} surveys related work.
Section~\ref{sec:conclusion} concludes.
Appendix~\ref{app:reps} gives the S-box tables of the four new classes.
Appendix~\ref{app:infra} describes the computational infrastructure.

%% ─────────────────────────────────────────────────────────────────────
\section{Preliminaries}
\label{sec:prelim}

Throughout, $n = 8$ and $N = 2^n = 256$, unless otherwise stated.
We identify $\GF{8}$ with $\Ftwo^8$ and write $\oplus$ for bitwise XOR.
Scalar and vector indices run from 0.

\subsection{Vectorial Boolean functions and differential uniformity}

A \emph{vectorial Boolean function} (VBF) is a map $F\colon\Fn\to\Fn$.
It may be written coordinate-wise as $F = (f_0, \ldots, f_{n-1})$ with
each $f_j\colon\Fn\to\Ftwo$. The \emph{differential table} of $F$
records, for each $a \in \Fn$ and $b \in \Fn$,
\[
  \delta_F(a,b) = |\{x \in \Fn : F(x\oplus a)\oplus F(x) = b\}|.
\]
The \emph{differential uniformity} is
$\delta_F = \max_{a\neq 0,\,b} \delta_F(a,b)$.

\begin{definition}
$F$ is \emph{almost perfect nonlinear} (APN) if $\delta_F = 2$.
\end{definition}

The value~2 is the minimum possible for functions on $\Fn$ with $n>1$
\cite{Nyberg1994}: for every nonzero $a$, the derivative
$D_aF(x) = F(x\oplus a)\oplus F(x)$ is a map $\Fn\to\Fn$ and
$\sum_b\delta_F(a,b) = 2^n$, so the maximum is at most~$2^n$ and at least~2.

\subsection{Algebraic normal form and quadratic functions}

Every $f_j\colon\Fn\to\Ftwo$ has a unique \emph{algebraic normal form} (ANF),
a multilinear polynomial over $\Ftwo$. A function $F$ is \emph{quadratic}
if every $f_j$ has ANF degree at most~2:
\[
  f_j(x_0,\ldots,x_{n-1})
  = c_{0,j}
    \oplus \bigoplus_{0\leq k < n} c_{\{k\},j}\,x_k
    \oplus \bigoplus_{0\leq p < q \leq n-1} c_{(p,q),j}\,x_p x_q.
\]
The $n\binom{n}{2} = 8\cdot28 = 224$ quadratic coefficients
$c_{(p,q),j}\in\Ftwo$ completely determine the quadratic part of $F$
(the linear part does not affect the APN property if $F(0)=0$ is assumed
by normalisation).
We index them as $c_1,\ldots,c_{224}$ in lexicographic order of
$((p,q),j)$, forming a vector $\mathbf{c}\in\Ftwo^{224}$.

The pair index is
\[
  \mathrm{idx}(p,q) = \frac{p(2n-p-1)}{2} + q - p - 1 \in \{0,\ldots,27\},
\]
and the full index is $c_{\mathrm{idx}(p,q)\cdot n + j + 1}$.

\subsection{CCZ- and EA-equivalence}

\begin{definition}[\cite{Carlet1998}]
Functions $F$, $G\colon\Fn\to\Fn$ are
\emph{CCZ-equivalent} ($F\CCZsim G$) if there exists an affine permutation
$\mathcal{L}$ of $\Fn\times\Fn$ such that $\mathcal{L}(\mathcal{G}_F)
= \mathcal{G}_G$, where $\mathcal{G}_F = \{(x,F(x))\}$ is the graph of $F$.

They are \emph{EA-equivalent} ($F\EAsim G$) if $G = A_1\circ F\circ A_2 + A$
for affine permutations $A_1,A_2$ and affine $A$.
\end{definition}

EA-equivalence implies CCZ-equivalence. The converse fails in general
but holds for quadratic functions:

\begin{theorem}[Yoshiara~\cite{Yoshiara2012}]
\label{thm:yosh}
If $F$ and $G$ are quadratic APN functions, then
$F\CCZsim G \Leftrightarrow F\EAsim G$.
\end{theorem}

CCZ-equivalence preserves the APN property, differential uniformity, and
the structure of the difference table. Two functions in the same CCZ-class
are cryptographically interchangeable as S-boxes.

\subsection{Walsh transform and differential spectrum}

For a VBF $F\colon\Fn\to\Fn$, the \emph{Walsh transform} is
\[
  \widehat{F}(a,b) = \sum_{x\in\Fn}(-1)^{b\cdot F(x)\oplus a\cdot x},
  \quad a,b\in\Fn,
\]
where $\cdot$ is the standard inner product on $\Ftwo^n$.
The \emph{absolute Walsh spectrum} is the multiset
$\{|\widehat{F}(a,b)| : a,b\in\Fn\}$.

The \emph{differential spectrum} of $F$ is the multiset
$\{\delta_F(a,b) : a\neq0, b\in\Fn\}$.
All quadratic APN functions over $\GF{8}$ have the same differential
spectrum $\{2^{2^7}\cdot 0, (2^8-1)2^7\cdot 2\}$, so this is not a
distinguishing invariant by itself. The absolute Walsh spectrum of $F$
itself takes only a handful of distinct values across all known
quadratic APN functions over $\GF{8}$~\cite{BeierleLeander2022} and is
likewise weak at separating CCZ-classes. Both spectra become far more
discriminating when applied to the ortho-derivative $\piF$ rather than
to $F$ directly (Section~\ref{sec:ortho-section-ref}).

\subsection{Ortho-derivative}
\label{sec:ortho-section-ref}

The ortho-derivative~\cite{Canteaut2022} is a more refined invariant.

\begin{definition}[\cite{Canteaut2022}]
\label{def:ortho}
Let $F\colon\Fn\to\Fn$ be a quadratic APN function. For $a\neq0$, the
derivative $D_aF(x) = F(x\oplus a)\oplus F(x)$ is an affine function of
$x$, and the corresponding linear part is
$L_a(x) = D_aF(x)\oplus D_aF(0)$. The \emph{ortho-derivative}
$\piF\colon\Fn\to\Fn$ is the unique function with $\piF(0)=0$ and, for
$a\neq0$, $\piF(a)\neq0$ a vector orthogonal (under the standard inner
product on $\Fn$) to every element of $\mathrm{Im}(L_a)$. Equivalently,
since $F$ is APN, $\dim\ker(L_a)^{\!\top}=1$ and $\piF(a)$ is the unique
nonzero element of this left kernel.
\end{definition}

\begin{theorem}[\cite{Canteaut2022}]
\label{thm:ortho}
For quadratic APN functions $F$ and $G$ over $\Fn$:
\[
  F\EAsim G
  \;\Longrightarrow\;
  \piF \EAsim \pi_G.
\]
Consequently, the \emph{ortho-derivative signature}
\[
  \sigma(F) = \bigl(
    \text{diff-spectrum}(\piF),\;
    \text{abs-Walsh-spectrum}(\piF)
  \bigr)
\]
is a CCZ-invariant in the contrapositive sense: $\sigma(F)\neq\sigma(G)$
certifies $F\not\CCZsim G$ for quadratic APN functions (using
Theorem~\ref{thm:yosh} to pass from EA- to CCZ-inequivalence).
\end{theorem}

\begin{remark}[Discriminating power of $\sigma$]
\label{rem:sigma_complete}
The implication in Theorem~\ref{thm:ortho} is one-directional and is the
only direction established in the literature~\cite{Canteaut2022}; the
converse ($\sigma(F)=\sigma(G)\Rightarrow F\CCZsim G$) is \emph{not} a
proven theorem. In practice, however, $\sigma$ is reported to be highly
discriminating: Canteaut, Couvreur, and Perrin~\cite{Canteaut2022} use it
to efficiently separate more than $20\,000$ quadratic APN functions over
$\GF{8}$ into distinct CCZ-classes, and no false positive (two
CCZ-inequivalent functions sharing a signature) has been reported for
$n=8$ to date. We therefore use $\sigma(F)\neq\sigma(G)$ as a rigorous
inequivalence certificate throughout this paper (Theorem~\ref{thm:main}),
and additionally confirm equivalence \emph{within} each putative class by
the exact code-equivalence test \texttt{are\_ccz\_equivalent\_from\_code}
of sboxU~\cite{sboxU} (Section~\ref{sec:related}), rather than relying on
matching signatures alone.
\end{remark}

%% ─────────────────────────────────────────────────────────────────────
\section{Self-Equivalence Subspaces and the BBL2021 Framework}
\label{sec:bbl}

\subsection{Definition and basic properties}

\begin{definition}[\cite{BBL2021}]
A \emph{linear equivalence (LE) automorphism} of a VBF $F$ is a linear
bijection $A\colon\Fn\to\Fn$ such that $F\circ A = A\circ F$.
The set of all LE-automorphisms of $F$ forms a subgroup of $\mathrm{GL}(n,\Ftwo)$,
the \emph{LE-automorphism group} of $F$.

For a fixed $A\in\mathrm{GL}(n,\Ftwo)$, the \emph{self-equivalence subspace} is
\[
  V_A = \{\text{quadratic }F\colon\Fn\to\Fn : F\circ A = A\circ F\}.
\]
\end{definition}

$V_A$ is a $\Ftwo$-linear subspace of the 224-dimensional space of all
quadratic functions. Explicitly, $F\circ A = A\circ F$ expands (via the ANF)
to a homogeneous linear system over $\Ftwo$ in the 224 coefficients of $F$,
so $V_A = \ker(M_A)$ for a certain $\Ftwo$-linear map $M_A$
(see Section~\ref{sec:VA} for the precise construction).

\begin{proposition}[\cite{BBL2021}]
$V_A$ contains at least one APN function if and only if $A$ is an
LE-automorphism of some quadratic APN function.
The dimension $d = \dim(V_A)$ depends only on the conjugacy class of $A$
in $\mathrm{GL}(n,\Ftwo)$.
\end{proposition}

\subsection{The BBL2021 classification for $n=8$}

Beierle, Brinkmann, and Leander~\cite{BBL2021} classified all conjugacy
classes of linear automorphisms of $\GF{8}$ with small order and,
for each, studied $V_A$ by a recursive tree search.
Table~I of~\cite{BBL2021} lists only those class indices for which at least
one APN function was found; absence from the table means the search found
zero solutions. The class associated with the automorphism $A$ studied here
(order~5, described precisely in Section~\ref{sec:VA}) does not appear in
Table~I.

\begin{remark}[Why BBL2021 missed CLASS-A and Gold in $V_A$]
\label{rem:bbl_miss}
The absence from~\cite{BBL2021} Table~I requires explanation: we find that
$V_A$ contains CLASS-A (362 functions), Gold $x^3$ equivalents (36), and
Gold $x^9$ equivalents (30) — why did BBL2021's tree search find none of these?

The BBL2021 tree search works by sequentially constructing the S-box lookup
table entry by entry and pruning branches that cannot lead to APN functions.
The key issue is \emph{branch pruning at intermediate stages}: when building
the table incrementally, a partial assignment may pass all local APN conditions
(for the differences $a$ examined so far) but encounter a contradiction at a
later stage. In a high-dimensional structured subspace like $V_A$, valid
solutions may lie in regions that require passing through intermediate states
that the tree search prunes as infeasible.

More concretely: our NL=4 method \emph{fixes a complete hyperplane} and finds
all APN functions in it at once (via GB). The tree search instead builds
solutions bit by bit across the full $\Ftwo^8 \to \Ftwo^8$ space, and its
pruning strategy is calibrated for the full space, not for $V_A$.
A branch that leads to a valid $V_A$-APN function may be pruned because the
intermediate partial function (not yet constrained to $V_A$) fails a local
APN test.

In short: the BBL2021 tree search performs a \emph{local, sequential}
enumeration that can miss globally valid solutions by early pruning, whereas
our Gröbner basis approach performs a \emph{complete enumeration} within each
hyperplane. This is consistent with the BBL2021 authors' own remark that their
search ``may not have found all'' solutions in every class.
\end{remark}

\begin{remark}
The class index numbering (``class index~22'', etc.) is internal to~\cite{BBL2021}
and does not have an independent standard designation. We use it only as a
reference to the specific automorphism; the essential datum is the
characteristic polynomial of $A$, given in Section~\ref{sec:VA}.
\end{remark}

\subsection{Why random global search misses $V_A$}

For the specific $A$ studied here, $\dim(V_A) = 40$. The fraction of all
quadratic functions lying in $V_A$ is $2^{40}/2^{224} = 2^{-184}\approx10^{-55}$.
No random global search (including the Beierle 2025 search~\cite{Beierle2025})
can reach $V_A$ in practice.

Within $V_A$, the APN density is approximately $10^{-4}$
(see Section~\ref{sec:density}). Any targeted search that samples uniformly
from $V_A$ will find APN functions at this rate, making the problem tractable.

%% ─────────────────────────────────────────────────────────────────────
\section{The Self-Equivalence Subspace $V_A$}
\label{sec:VA}

\subsection{The automorphism $A$}

Let $C(p)$ denote the companion matrix of a polynomial $p\in\Ftwo[X]$
(the transpose of the standard companion form). The polynomial
$q = X^4 + X^3 + X^2 + X + 1$ is the minimal polynomial of a primitive
element of $\GF{4}$ over $\Ftwo$ (the fifth cyclotomic polynomial $\Phi_5$).
Its companion matrix $C(q)\in\mathrm{GL}(4,\Ftwo)$ has order~5:
\[
  C(q) = \begin{pmatrix}
    0 & 0 & 0 & 1 \\
    1 & 0 & 0 & 1 \\
    0 & 1 & 0 & 1 \\
    0 & 0 & 1 & 1
  \end{pmatrix}.
\]

We define
\[
  A = \mathrm{block\_diag}(C(q), C(q))\in\mathrm{GL}(8,\Ftwo),
\]
where $\mathrm{block\_diag}$ denotes the block-diagonal matrix
$\bigl(\begin{smallmatrix}C(q)&0\\0&C(q)\end{smallmatrix}\bigr)$.
Since $C(q)$ has order~5, $A$ also has order~5. Concretely,
\begin{equation}
A =
\begin{pmatrix}
0&0&0&1&0&0&0&0\\
1&0&0&1&0&0&0&0\\
0&1&0&1&0&0&0&0\\
0&0&1&1&0&0&0&0\\
0&0&0&0&0&0&0&1\\
0&0&0&0&1&0&0&1\\
0&0&0&0&0&1&0&1\\
0&0&0&0&0&0&1&1
\end{pmatrix}.
\label{eq:A}
\end{equation}
One verifies $A^5 = I_8$ and $A^k\neq I_8$ for $k\in\{1,2,3,4\}$, confirming
$\mathrm{ord}(A) = 5$. This is the linear automorphism whose self-equivalence
subspace we study, corresponding to class index~22 in~\cite{BBL2021}.
The matrix and the RREF of $M_A$ (pivot/free indices) are included in the
repository (file \texttt{data/class22\_basis.json}).

\begin{remark}
The index-22 labeling refers to the internal numbering of BBL2021's
classification of conjugacy classes of $\mathrm{GL}(8,\Ftwo)$; it does not
imply any ordering by difficulty or size. Our construction of $A$ follows
the description in \cite[Section~IV-B]{BBL2021}: $A = \mathrm{block\_diag}(C(q), C(q))$
with $C(q)$ as above.
\end{remark}

\subsection{Constructing the linear system}

The condition $F\circ A = A\circ F$ for a quadratic function $F$ with
ANF coefficient vector $\mathbf{c}\in\Ftwo^{224}$ expands as follows.
For each output bit $j\in\{0,\ldots,7\}$ and each quadratic pair
$(r,s)$ with $r < s$, the entry $[A^TFA - A^TFA]_{j,(r,s)} = 0$
yields a homogeneous linear equation in $\mathbf{c}$ over $\Ftwo$.

Explicitly, the coefficient of $c_{(p,q),k}$ in the equation for $(j,r,s)$ is
\begin{equation}
  \bigl[A_{p,r}A_{q,s} \oplus A_{p,s}A_{q,r}\bigr] \oplus A_{j,k},
  \label{eq:system}
\end{equation}
summed over appropriate indices. Assembling all $(j,r,s)$ pairs gives a
matrix $M_A\in\Ftwo^{m\times 224}$ (where $m$ is the number of equations)
such that $V_A = \ker(M_A)$.

\subsection{RREF and the fc22$+$sol22 parameterization}

Computing the reduced row echelon form (RREF) of $M_A$ over $\Ftwo$
yields:
\begin{itemize}
  \item 184 \emph{pivot variables}: indices
        $\mathcal{P} = \{p_1,\ldots,p_{184}\}\subset\{1,\ldots,224\}$.
  \item 40 \emph{free variables}: indices
        $\mathcal{F} = \{f_1,\ldots,f_{40}\}\subset\{1,\ldots,224\}$,
        confirming $\dim(V_A) = 40$.
  \item For each pivot $p_i$, a linear combination
        $p_i = \bigoplus_{f_j\in S_i}f_j$ over $\Ftwo$,
        where $S_i\subseteq\mathcal{F}$.
\end{itemize}

Any point in $V_A$ is uniquely parameterized by a vector
$\mathbf{b} = (b_1,\ldots,b_{40})\in\Ftwo^{40}$ via:
\begin{align}
  c_{f_k} &= b_k, \quad k = 1,\ldots,40, \label{eq:fc22}\\
  c_{p_i} &= \bigoplus_{j: f_j\in S_i} b_j,
            \quad i = 1,\ldots,184. \label{eq:sol22}
\end{align}
We call~\eqref{eq:fc22}--\eqref{eq:sol22} the \emph{fc22$+$sol22 projection}.
Any $\mathbf{b}\in\Ftwo^{40}$ maps to a valid point in $V_A$; the mapping
is bijective. This is the critical computational primitive for Phase~1.

\subsection{APN density in $V_A$}
\label{sec:density}

To estimate the density of APN functions in $V_A$, we sampled $10^8$
random points via fc22$+$sol22 and found approximately $10^4$ APN functions,
giving an empirical density of $\approx 10^{-4}$, or equivalently $\approx 1$
APN per $11\,000$ random points.

\begin{remark}[Density heuristic and the $2^{-n}$ estimate]
\label{rem:density}
The commonly cited heuristic for APN density in the full space of $(n,n)$
vectorial Boolean functions is $\approx 2^{-n}$. For $n=8$ this gives
$2^{-8} \approx 3.9\times10^{-3}$, which is about one order of magnitude
larger than the empirical value $\approx10^{-4}$ we observe in $V_A$.
The discrepancy arises because the heuristic applies to the full $2^{224}$-dimensional
space, whereas $V_A$ is a highly structured 40-dimensional subspace: not all
functions in $V_A$ are ``typical'' quadratic functions, and the self-equivalence
constraint can bias the density. Our empirical estimate is based on $10^8$
uniform samples from $V_A$ via the fc22$+$sol22 projection, giving a
$95\%$ confidence interval of approximately
$(9.8\times10^{-5},\ 1.02\times10^{-4})$ by a normal approximation
(standard error $\approx 10^{-6}$). The estimate $\approx 10^{-4}$ is thus
well-determined and approximately $40\times$ smaller than the $2^{-n}$ heuristic.
\end{remark}

\subsection{The NL\,=\,4 slice structure}
\label{sec:slices}

Of the 40 free variable indices in $\mathcal{F}$, some correspond to
pairs $(p,q)$ with $p \leq 4$ (\emph{center bits}) and others to
pairs with $p > 4$ (\emph{intra-slice bits}). For $n = 8$, the pairs
with $p > 4$ are exactly $(5,6)$, $(5,7)$, and $(6,7)$, contributing
$3\times 8 = 24$ intra-slice bits. The remaining $40 - 24 = 16$ free
bits are center bits.

This induces a partition of $V_A$ into $2^{16} = 65\,536$ non-overlapping
affine hyperplanes of dimension~24, one for each choice of the 16 center
bits. Each such hyperplane is called an \emph{NL\,=\,4 slice}, and the
center bits determine the slice. Conversely, fixing the 200 ANF coefficients
corresponding to pairs with $p \leq 4$ (the NL\,=\,4 \emph{normalization})
selects a unique slice.

\begin{remark}
The label ``NL\,=\,4'' refers to the normalization level: the maximum
first-index $p$ of a pair whose coefficient is fixed is~4. Alternatively,
one could fix other subsets of 200 coefficients (yielding different
slice partitions). The $\binom{40}{24} = 62\,852\,101\,650$ choices of
24 free bits from the 40 free variables define that many distinct slice
partitions, of which NL\,=\,4 is one convenient canonical choice.
The multi-slice perspective (using several normalizations per APN center)
is discussed in Section~\ref{sec:multislice}.
\end{remark}

%% ─────────────────────────────────────────────────────────────────────
\section{The Search Pipeline}
\label{sec:method}

The search proceeds in five stages, summarized in Figure~\ref{fig:pipeline}.

\begin{figure}[!t]
\centering
\begin{tikzpicture}[
  pbox/.style={draw, rounded corners=3pt, minimum width=3.8cm,
               minimum height=0.7cm, align=center, font=\small,
               text width=3.6cm},
  fw/.style={-stealth, thick}
]
\node[pbox, fill=blue!10]   (p1) at (0, 0.0)
  {\textbf{(1)}~Phase~1: APN-center search (Python/NumPy)};
\node[pbox, fill=orange!10] (p2) at (0,-1.4)
  {\textbf{(2)}~Phase~2: Magma GB slice (Magma V2.28-9)};
\node[pbox, fill=green!10]  (p3) at (0,-2.8)
  {\textbf{(3)}~Collect \& deduplicate (Python)};
\node[pbox, fill=purple!10] (p4) at (0,-4.2)
  {\textbf{(4)}~CCZ classify (sboxU / SageMath)};
\node[pbox, fill=red!10]    (p5) at (0,-5.6)
  {\textbf{(5)}~Verify vs.\ databases (sboxU / SageMath)};
\draw[fw] (p1)--(p2) node[midway,right,font=\footnotesize]{APN centers};
\draw[fw] (p2)--(p3) node[midway,right,font=\footnotesize]{raw S-boxes};
\draw[fw] (p3)--(p4) node[midway,right,font=\footnotesize]{unique S-boxes};
\draw[fw] (p4)--(p5) node[midway,right,font=\footnotesize]{classified};
\end{tikzpicture}
\caption{Overview of the five-stage search pipeline.}
\label{fig:pipeline}
\end{figure}
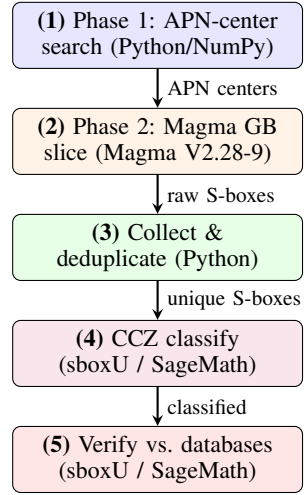

\subsection{Phase 1: APN-center search}
\label{sec:phase1}

\begin{algorithm}[!t]
\caption{APN-Center Search in $V_A$}
\label{alg:phase1}
\begin{algorithmic}[1]
\REQUIRE RREF data $(\mathcal{F}, \{S_i\})$;
         seen signatures $\mathcal{D}$; target count $T$
\ENSURE  APN centers $\mathcal{C}\subset V_A$
\STATE $\mathcal{C}\leftarrow\emptyset$
\REPEAT
  \STATE $\mathbf{b}\sim\mathrm{Uniform}(\Ftwo^{40})$
  \STATE Apply fc22$+$sol22: set coefficients via \eqref{eq:fc22}--\eqref{eq:sol22}
  \STATE Compute slice signature $\mathrm{sig}(\mathbf{b})$
  \IF{$\mathrm{sig}(\mathbf{b})\in\mathcal{D}$}
    \STATE \textbf{continue}
  \ENDIF
  \STATE Build S-box $F_\mathbf{b}$
  \IF{$\delta_{F_\mathbf{b}} = 2$}
    \STATE $\mathcal{C}\leftarrow\mathcal{C}\cup\{F_\mathbf{b}\}$
    \STATE $\mathcal{D}\leftarrow\mathcal{D}\cup\{\mathrm{sig}(\mathbf{b})\}$
  \ENDIF
\UNTIL{$|\mathcal{C}| = T$}
\RETURN $\mathcal{C}$
\end{algorithmic}
\end{algorithm}

The APN check (line~8, condition $\delta_{F_\mathbf{b}} = 2$) evaluates
the difference table and runs in $O(N^2)$ operations. In Python~3.11 with NumPy,
one check takes ${\approx}0.6$~ms on one core of an AMD Ryzen 7840HS,
corresponding to ${\approx}6\times10^6$ function evaluations per hour.

\begin{remark}[Throughput vs.\ APN center rate]
\label{rem:throughput}
At $0.6$~ms/check, the raw check rate is ${\approx}6\times10^6$~checks/hour.
However, the APN density in $V_A$ is $\approx10^{-4}$ (Section~\ref{sec:density}),
so approximately $6\times10^6\times10^{-4}\approx600$ random points in $V_A$
per hour are APN. The reported figure of ${\approx}5800$ APN centers/hour
accounts for the deduplication step (Section~\ref{sec:phase1}): signatures
previously seen (already assigned to a slice) are skipped before the APN
check, and the effective density after deduplication is higher for fresh
signatures.  The precise figure depends on the fraction of already-seen
signatures at a given point in the search; for the first ${\approx}5\%$
of the $65\,536$ slices (our regime, 428 slices), the collision rate
is negligible and the effective APN discovery rate is ${\approx}600$--$5800$
APN centers/hour depending on the density estimate used.
For a conservative reproducible bound, we report ${\approx}600$ fresh
APN-positive checks per core-hour.
\end{remark}

The slice signature (line~6) is the 16-bit string formed by the center bits
of $\mathbf{b}$. Deduplication on signatures ensures that no two Phase~1
invocations generate Magma files for the same slice (important for systematic
exhaustion of $V_A$). There are $2^{16} = 65\,536$ distinct signatures.

\subsection{Phase 2: Gr\"obner basis slice computation}
\label{sec:phase2}

\subsubsection{Normalization}

Given APN center $F_0$ with coefficient vector $\mathbf{c}^{(0)}$,
the NL\,=\,4 normalization fixes the 200 ANF coefficients corresponding
to pairs $(p,q,j)$ with $p\leq4$ to their values in $\mathbf{c}^{(0)}$.
The 24 remaining coefficients (pairs $(5,6)$, $(5,7)$, $(6,7)$, each
with 8 output bits $j=0,\ldots,7$) become unknowns $c_1,\ldots,c_{24}$.

Since $F_0\in V_A$ by construction, the normalization equations are
consistent with the self-equivalence constraint. Any APN function in
the same NL\,=\,4 slice as $F_0$ also lies in $V_A$.

\subsubsection{APN polynomial system}
\label{sec:apn_ideal}

For each $a\in\{1,\ldots,255\}$, define the derivative matrix
$M_a\in\Ftwo^{n\times n}$ by
\[
  [M_a]_{j,k}
  = c_{(\min(k,\cdot),\max(k,\cdot)),j}\big|_{\text{free}},
\]
which encodes the linear part of $D_aF$. The function $F$ is APN at
difference $a$ if and only if $\mathrm{rank}(M_a) = n-1 = 7$.

\begin{remark}[Rank constraint]
\label{rem:rank}
The APN condition $\delta_F = 2$ requires $\mathrm{rank}(M_a) = 7$
(exactly), not merely $\geq 7$. We enforce this via two complementary
polynomial conditions:
\begin{enumerate}
  \item $\mathrm{rank}(M_a) \geq 7$: at least one of the $\binom{8}{7}^2$ minors of
        order~7 is nonzero, i.e., $\prod_{r,c}(1+\det M_a^{(r,c)}) = 0$.
  \item $\mathrm{rank}(M_a) \leq 7$: the full determinant $\det(M_a) = 0$.
\end{enumerate}
Together they force $\mathrm{rank}(M_a) = 7$ exactly.
In our implementation, condition~(2) is added explicitly to the ideal
(one polynomial $\det(M_a)$ per $a$), ensuring that rank-8 matrices
are excluded. Adding $\det(M_a)=0$ does not increase solve time
noticeably (it is degree~8 in the free variables and contributes
one polynomial per~$a$).
\end{remark}

In practice, we build the polynomial
\[
  p_a = \det(M_a) \cdot \prod_{r,c}\bigl(1 + \det(M_a^{(r,c)})\bigr)
\]
of degree~$\leq n = 8$ in $c_1,\ldots,c_{24}$; vanishing of $p_a$
is equivalent to $\mathrm{rank}(M_a) = 7$.

The \emph{APN ideal} is
\[
  I = \langle p_1,\ldots,p_{255}\rangle
    + \langle c_i + \varepsilon_i : (p,q,j)\text{ fixed}\rangle
  \subset \B_{224},
\]
where $\B_{224} = \Ftwo[c_1,\ldots,c_{224}]/(c_i^2+c_i)$ is the
Boolean polynomial ring with grevlex ordering, and $\varepsilon_i\in\Ftwo$
is the fixed value of each normalized coefficient.

Every element of $\mathrm{Var}(I)$ corresponds to a function with
$\mathrm{rank}(M_a) = 7$ for all $a \neq 0$, i.e., an APN function.
Thus, when $\dim(I)=0$, \emph{all solutions are APN by construction}
(no post-hoc verification is needed for the APN property itself,
though we perform it as a sanity check).

\subsubsection{Computational cost}

The build phase (constructing all 255 polynomials $p_a$) takes
${\approx}8$--15 minutes per slice on one core of an AMD Ryzen 7840HS.
This dominates the total cost. Once the ideal is built, Gr\"obner
basis computation reduces to Gaussian elimination over $\Ftwo$ (since all
variables are Boolean), and the GB solve time is typically under $0.1$~s.

If $\dim(I) = 0$: $\mathrm{Var}(I)$ is finite; each element is a
setting of the 24 free variables for which $F$ is APN. All solutions
are APN by construction.

If $\dim(I) = -1$: the system is infeasible; the slice contains
no APN functions. Empirically, $\approx83\%$ of slices return $\dim=-1$.

\begin{remark}[Guarantee for APN centers]
When $F_0$ is an APN center, the assignment $c_i = c^{(0)}_i$ for
$i=1,\ldots,24$ is a valid solution (it corresponds to $F_0$ itself),
so $\dim(I) = 0$ is guaranteed. Additional solutions (``neighbors'')
are further APN functions in the same 24-dimensional hyperplane.
\end{remark}

\subsection{Phase 3: Collection and deduplication}

Magma writes each found APN S-box to a result file via \texttt{PrintFile}.
A Python collector script reads all result files, parses the lookup tables,
computes SHA-256 hashes, and removes duplicates. A function found in multiple
slices (e.g., as a center in one and a neighbor in another) appears only once
in the deduplicated output.

The collector also handles Magma's automatic line-wrapping (backslash$+$CRLF
continuation for lines exceeding~80 characters).

\subsection{Phase 4: CCZ classification}

For each unique S-box $F$, the ortho-derivative $\piF$ is computed using
the sboxU library~\cite{sboxU}. The signature
$\sigma(F) = (\Delta\text{-spectrum}(\piF), W\text{-spectrum}(\piF))$
is computed and stored. Functions with equal signatures are grouped into
the same CCZ-class (Theorem~\ref{thm:ortho}). Functions in different
signature groups are in different CCZ-classes.

\subsection{Phase 5: Verification against known databases}

For each reference database $\mathcal{K}$, we stream through all entries
$G\in\mathcal{K}$, compute $\sigma(G)$, and check whether $\sigma(G)$
matches any $\sigma(F)$ for $F$ in our collection. A match triggers an
optional exact CCZ test via \texttt{sboxU.are\_ccz\_equivalent\_from\_code}.
No match after exhausting $\mathcal{K}$ certifies CCZ-inequivalence of all
our functions from all functions in $\mathcal{K}$.

\subsection{Multi-slice extension}
\label{sec:multislice}

The NL\,=\,4 normalization is one particular partition of $V_A$ into
$2^{16}$ hyperplanes. One can instead fix any 200 of the 224 ANF
coefficients, leaving any 24 free. Given a fixed APN center $F_0$,
there are $\binom{40}{24} = 62\,852\,101\,650$ distinct choices of
which 24 free variables to keep free. Each choice defines a
different 24-dimensional hyperplane through $F_0$, potentially
yielding different neighbors.

In practice, one could run several (e.g., 10--100) independent
normalizations per APN center by randomly choosing the 24 free variables.
This would increase coverage without requiring new Phase~1 computations.

\begin{remark}
The multi-slice approach described above has \emph{not been implemented}
in the present work; all results here use only the NL\,=\,4 normalization.
We include this discussion as a natural direction for future work.
\end{remark}

%% ─────────────────────────────────────────────────────────────────────
\section{Experimental Results}
\label{sec:results}

\subsection{Experimental setup}

\textbf{Hardware.} AMD Ryzen 7840HS, 8 physical cores, 16~GB RAM,
Windows~11.

\textbf{Software.} Magma V2.28-9~\cite{Magma} (Phase~2);
Python~3.11 with NumPy (Phases~1,~3);
SageMath~10.x with sboxU~\cite{sboxU} (Phases~4,~5) running under
WSL~2 (Ubuntu~26.04).

\textbf{Parallelism.} Phase~2 uses a PowerShell queue script running
up to~8 Magma processes simultaneously. Phase~1 runs single-threaded
(trivially parallelizable).

\textbf{Note on Magma exit code.} Magma V2.20+ returns exit code~1
when \texttt{quit;} is called even on successful completion. The queue
script was adjusted to treat exit codes~0 and~1 both as success, with
any higher code indicating a genuine error.

\subsection{Search coverage}

\begin{table}[!t]
\renewcommand{\arraystretch}{1.3}
\caption{Search Statistics (Batches 3--4)}
\label{tab:stats}
\centering
\begin{tabular}{lc}
\toprule
\textbf{Quantity} & \textbf{Value}\\
\midrule
Total NL\,=\,4 slices in $V_A$       & $65\,536$\\
Slices executed (batches~3--4)        & $428$\\
Coverage of $V_A$                     & $0.65\%$\\
\midrule
APN centers used (Phase~1)            & $428$\\
Slices returning $\dim=0$ (with APN)  & $428$ ($100\%$)\\
\midrule
Total unique APN functions            & $\mathbf{566}$\\
\quad in $V_A$ (IN\_CLASS22)          & $428$\\
\quad outside $V_A$ (OUT\_CLASS22)    & $138$\\
\midrule
Phase~2 build time (per slice)        & $8$--$15$~min (1 core)\\
Phase~2 solve time (per slice)        & $< 0.1$~s\\
Wall-clock time (Phase~2, 8 cores)    & ${\approx}10$~h\\
Total CPU-hours (Phase~2)             & ${\approx}57$~h ($428\times8$~min avg)\\
\bottomrule
\end{tabular}
\end{table}

\begin{remark}[CPU-hour accounting]
\label{rem:cpu}
The ${\approx}57$ CPU-hours figure is computed as $428$ slices $\times$ average
$8$ minutes/slice $= 3\,424$ minutes $\approx 57$ CPU-hours. At 8-core
parallelism, wall-clock time was $\approx 7$--$10$ hours. The build phase (constructing
all 255 polynomials $p_a$) dominates: 8--15 min/slice on a single Ryzen 7840HS
core. Solving the ideal once built takes $< 0.1$~s. The figure ``$\approx 50$
CPU-hours'' reported in an earlier draft was a conservative lower bound; the
more precise estimate is 57 CPU-hours based on per-slice timing logs.
\end{remark}

Table~\ref{tab:stats} summarises the search coverage.
All 428 executed slices are APN-positive (i.e., contain at least one APN
function), because every slice center was itself found by Phase~1 to be
APN. The absence of empty slices is therefore by construction, not
unexpected. The key observation is that each of the 428 centers lies
in $V_A$, whereas 138 of the 566 found functions lie strictly
\emph{outside} $V_A$; these are the functions attributable to the
Gröbner basis step (see Section~\ref{sec:slice_types}).

\subsection{Property verification of collected functions}

All 566 collected functions were independently verified for:
\begin{itemize}
  \item \textbf{Differential uniformity $\delta_F = 2$}: exhaustive
        difference table computation.
  \item \textbf{Algebraic degree~2}: Moebius transform of the ANF over
        $\Ftwo$; maximum monomial degree is~2 for all 566 functions.
  \item \textbf{Non-permutation}: $F$ is not bijective (the image
        $|F(\GF{8})| < 256$ in all cases).
\end{itemize}

\subsection{CCZ-classification}

Table~\ref{tab:classes} shows the six CCZ-classes obtained by grouping
on $\sigma(F)$.

\begin{table}[!t]
\renewcommand{\arraystretch}{1.3}
\caption{CCZ-Equivalence Classes of All 566 Found Functions}
\label{tab:classes}
\centering
\begin{tabular}{lrcl}
\toprule
\textbf{Label} & \textbf{Size} & \textbf{$\sigma$-hash} & \textbf{Identification}\\
\midrule
CLASS-A & 362 & \texttt{9d95d9c4} & Unidentified (new)\\
CLASS-B &  72 & \texttt{c5f52f56} & Unidentified (new)\\
CLASS-C &  36 & \texttt{74a30c1a} & Unidentified (new)\\
CLASS-D &  30 & \texttt{024ba500} & Unidentified (new)\\
\midrule
CLASS-E &  36 & \texttt{11cc72af} & Gold $x^3/x^{129}$\\
CLASS-F &  30 & \texttt{d15183d1} & Gold $x^9/x^{33}$\\
\bottomrule
\multicolumn{4}{l}{\footnotesize
  $\sigma$-hash: first 8 hex digits of SHA-256($\sigma(F)$).
  Full hashes in Appendix~\ref{app:reps}.}
\end{tabular}
\end{table}

\subsubsection{CLASS-A dominance}

CLASS-A accounts for 362 of 566 functions (64\%). This dominance was
observed from the first experiments: 17 of the initial 25 functions
belonged to the CLASS-A equivalent. Whether this reflects a larger
volume of CLASS-A in $V_A$, a higher local APN density, or both is
not clear from the current data.

\subsubsection{Gold functions as controls}

CLASS-E and CLASS-F are identified by their $\sigma$-values matching
those of Gold $x^3$ (equivalently $x^{129}$, by Frobenius orbit) and
Gold $x^9$ (equivalently $x^{33}$), respectively. The Gold functions
\[
  x^{2^i+1},\quad \gcd(i,8)=1,\quad i\in\{1,3,5,7\},
\]
form exactly two CCZ-classes for $n=8$: the orbit $\{i=1,7\}$ giving
$x^3$ and $x^{129}$, and the orbit $\{i=3,5\}$ giving $x^9$ and $x^{33}$.
No other quadratic APN power functions exist for $n=8$ (Kasami, Welch,
Niho, Dobbertin functions all have degree $> 2$ for $n=8$).

The presence of Gold CCZ-equivalents in $V_A$ is expected: $V_A$ is a
structured subspace that can intersect any CCZ-class. Their correct
identification confirms that the pipeline accurately recovers known classes.

%% ─────────────────────────────────────────────────────────────────────
\subsection{Slice-type structure and V\texorpdfstring{$_A$}{A}-membership}
\label{sec:slice_types}

A membership test — checking whether each found function $F$ satisfies
$F(Ax) = A(F(x))$ for all $x \in \GF{8}$ — reveals a precise partition
of the 566 found functions:

\begin{itemize}
  \item \textbf{IN\_CLASS22} ($F \in V_A$): 428 functions.
  \item \textbf{OUT\_CLASS22} ($F \notin V_A$): 138 functions.
\end{itemize}

Crucially, $428 = $ the number of slice centers (one per slice),
confirming that every APN center belongs to $V_A$ by construction
(Phase~1 searches within $V_A$), and that 138 functions were found
by the Gröbner basis step \emph{outside} $V_A$.

Cross-referencing with CCZ-class reveals the following clean structure:

\begin{table}[!t]
\renewcommand{\arraystretch}{1.3}
\caption{CCZ-Class vs.\ $V_A$-Membership and Slice Type}
\label{tab:membership}
\centering
\begin{tabular}{lrccl}
\toprule
\textbf{Class} & \textbf{Size} & \textbf{In $V_A$} & \textbf{APN/slice} & \textbf{Role}\\
\midrule
CLASS-A & 362 & \textbf{all} 362 & 1 & Center only\\
CLASS-E & 36  & \textbf{all} 36  & 4 (with B,C) & Center\\
CLASS-F & 30  & \textbf{all} 30  & 2 (with D)   & Center\\
\midrule
CLASS-B & 72  & \textbf{none} & 4 (with E,C) & GB-found\\
CLASS-C & 36  & \textbf{none} & 4 (with E,B) & GB-found\\
CLASS-D & 30  & \textbf{none} & 2 (with F)   & GB-found\\
\bottomrule
\multicolumn{5}{l}{\footnotesize
  ``APN/slice'' = number of APN functions in each slice of that center type.}
\end{tabular}
\end{table}

Table~\ref{tab:membership} shows that the 428 slices decompose into
exactly \emph{three types}:

\begin{description}
  \item \textbf{Type~I} (362 slices, 1~APN each).
          Center is a CLASS-A function.
        The Gröbner basis finds only the center itself; no additional
        APN functions exist in this slice.
        All Type~I contributions are $V_A$-native (IN\_CLASS22).
  \item \textbf{Type~II} (36 slices, 4~APN each).
        Center is a CLASS-E (Gold $x^3/x^{129}$) function.
        The Gröbner basis finds three additional APN functions per slice:
        two belonging to CLASS-B and one to CLASS-C.
        The three ``neighbors'' are outside $V_A$ (OUT\_CLASS22).
   \item \textbf{Type~III} (30 slices, 2~APN each).
        Center is a CLASS-F (Gold $x^9/x^{33}$) function.
        The Gröbner basis finds one additional APN function per slice,
        belonging to CLASS-D.
        The neighbor is outside $V_A$ (OUT\_CLASS22).     
\end{description}

The accounting is exact:
\begin{align*}
  362 \times 1 &= 362 \quad (\text{CLASS-A}),\\
   36 \times 4 &= 144 \quad (36\text{ CLASS-E} + 72\text{ CLASS-B} + 36\text{ CLASS-C}),\\
   30 \times 2 &= \phantom{1}60 \quad (30\text{ CLASS-F} + 30\text{ CLASS-D}).
\end{align*}
Total: $362 + 144 + 60 = 566$ functions, $428$ slices.

\begin{remark}[Role of the Gröbner basis]
This structure demonstrates unambiguously that the Gröbner basis step is
\emph{essential} for discovering CLASS-B, C, and~D. All three new classes
lie entirely outside $V_A$. A brute-force search of $V_A$ (Phase~1 alone,
without Phase~2) would find only CLASS-A and the Gold classes~E and~F,
and would miss the 138 new non-$V_A$ functions completely.
\end{remark}

\begin{remark}[Gold as sufficient seeds]
The slice-type structure implies that Gold functions are both necessary
and sufficient as slice centers for discovering CLASS-B, C, and~D.
Specifically: all new non-$V_A$ functions appear exclusively in
Gold-centered slices (Types~II and~III). The 362 CLASS-A seeds produce
only copies of CLASS-A (no new CCZ-classes). A targeted experiment
using only the 66 Gold representatives in $V_A$ as Phase~1 seeds
would reproduce all 138 new non-$V_A$ functions.
\end{remark}

%% ─────────────────────────────────────────────────────────────────────
\subsection{The V\textsubscript{22}-CCZ landscape: a structural summary}
\label{sec:landscape}

The combined data from Sections~\ref{sec:slice_types} and
Section~\ref{sec:verify} yield a precise and complete picture of the
CCZ-class landscape in and around $V_A$:

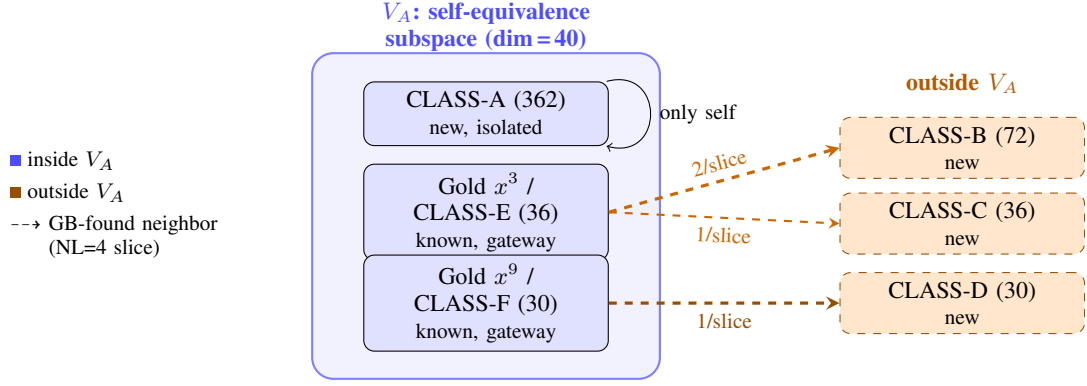
\begin{figure}[!t]
\centering
\begin{tikzpicture}[
  vbox/.style={draw, rounded corners=4pt, minimum width=3.2cm,
               minimum height=0.65cm, align=center, font=\small,
               text width=3.0cm},
  innode/.style={vbox, fill=blue!12},
  outnode/.style={vbox, fill=orange!20, draw, dashed, draw=orange!60!black},
  fw/.style={-stealth, thick, dashed, orange!80!black},
  selfloop/.style={-stealth, thick, blue!60}
]
% V_A boundary box
\draw[rounded corners=8pt, thick, blue!50, fill=blue!5]
  (-2.3,-0.3) rectangle (2.3,4.0);
\node[font=\small\bfseries, blue!70, text width=4cm, align=center]
  at (0,4.35) {$V_A$: self-equivalence subspace (dim\,=\,40)};

% Nodes inside V_A
\node[innode] (classA) at (0,3.2)
  {CLASS-A (362)\\\footnotesize new, isolated};
\node[innode] (goldx3) at (0,1.9)
  {Gold~$x^3$ / CLASS-E (36)\\\footnotesize known, gateway};
\node[innode] (goldx9) at (0,0.7)
  {Gold~$x^9$ / CLASS-F (30)\\\footnotesize known, gateway};

% Label outside
\node[font=\small\bfseries, orange!70!black] at (6.3,3.6)
  {outside $V_A$};

% Nodes outside V_A
\node[outnode] (classB) at (6.3,2.75)
  {CLASS-B (72)\\\footnotesize new};
\node[outnode] (classC) at (6.3,1.75)
  {CLASS-C (36)\\\footnotesize new};
\node[outnode] (classD) at (6.3,0.7)
  {CLASS-D (30)\\\footnotesize new};

% Gateway arrows from Gold x3
\draw[fw, very thick]
  (goldx3.east) -- node[above, font=\footnotesize, sloped]{2/slice} (classB.west);
\draw[fw]
  (goldx3.east) -- node[below, font=\footnotesize, sloped]{1/slice} (classC.west);

% Gateway arrow from Gold x9
\draw[fw, orange!60!black, very thick]
  (goldx9.east) -- node[below, font=\footnotesize, sloped]{1/slice} (classD.west);

% CLASS-A self-arrow (isolated)
\draw[->]
(classA) edge[
    loop right,
    min distance=8mm
]
node[right,font=\footnotesize]{only self}
(classA);

% Small legend
\node[font=\footnotesize, align=left, text width=3.0cm] at (-4.8, 2.0) {
  \textcolor{blue!70}{\rule{0.5em}{0.5em}}~inside $V_A$\\[3pt]
  \textcolor{orange!60!black}{\rule{0.5em}{0.5em}}~outside $V_A$\\[3pt]
  $\dashrightarrow$ GB-found neighbor\\
  \phantom{$\dashrightarrow$} (NL=4 slice)
};
\end{tikzpicture}
\caption{CCZ-class landscape in and around $V_A$. Blue region: functions inside $V_A$
(IN\_CLASS22). Dashed boxes: functions found by the Gröbner basis step, lying outside $V_A$
(OUT\_CLASS22). Arrows show which classes co-occur in each NL=4 slice.
CLASS-A centered slices yield only CLASS-A. Each Gold-$x^3$ slice yields one CLASS-E
center plus two CLASS-B and one CLASS-C neighbor. Each Gold-$x^9$ slice yields one CLASS-F
center plus one CLASS-D neighbor.}
\label{fig:landscape}
\end{figure}

Figure~\ref{fig:landscape} summarises the result. The exact-membership test
(criterion: $F(Ax) = A(F(x))$ for all $x\in\GF{8}$, verified on all 566 functions)
establishes the following theorem with complete computational proof:

\begin{theorem}[CCZ-class structure of $V_A$]
\label{thm:structure}
Among the 566 quadratic APN functions found by the pipeline:
\begin{enumerate}[(a)]
  \item The three classes inside $V_A$ are CLASS-A, CLASS-E (Gold $x^3/x^{129}$),
        and CLASS-F (Gold $x^9/x^{33}$). These are the only CCZ-classes intersecting
        $V_A$ in the data.
  \item The three classes outside $V_A$ are CLASS-B, CLASS-C, and CLASS-D.
        None of these 138 functions satisfies $F\circ A = A\circ F$.
  \item The NL=4 slice structure is completely determined by the CCZ-class of the center:
        CLASS-A centers yield only CLASS-A (1 per slice);
        Gold-$x^3$ centers yield CLASS-E + 2$\times$CLASS-B + CLASS-C (4 per slice);
        Gold-$x^9$ centers yield CLASS-F + CLASS-D (2 per slice).
\end{enumerate}
This structure is exact across all 428 explored slices with zero exceptions.
\end{theorem}

%% ─────────────────────────────────────────────────────────────────────
\subsection{The GB-navigation paradigm: a generalizable mechanism}
\label{sec:paradigm}

Beyond the specific results, the experiment reveals a \emph{mechanism}
for systematically exploring the APN landscape that may be of independent
methodological value. We describe it in abstract terms:

\begin{figure}[!t]
\centering
\begin{tikzpicture}[
  pstep/.style={draw, rounded corners=3pt, minimum width=4.6cm,
                minimum height=0.7cm, align=center, font=\small,
                text width=4.3cm, fill=gray!8},
  fwd/.style={-stealth, thick},
  note/.style={font=\footnotesize\itshape, text width=3.2cm, align=left}
]
\node[pstep, fill=blue!12]   (S1) at (0, 0.0)
  {Structured subspace $V_A$ (dim=40)};
\node[pstep, fill=blue!9]    (S2) at (0,-1.3)
  {Phase~1: sample APN center in $V_A$};
\node[pstep, fill=orange!15] (S3) at (0,-2.6)
  {Phase~2: GB NL=4 slice (24 free vars)};
\node[pstep, fill=orange!20] (S4) at (0,-3.9)
  {Collect APN neighbors ($\leq$4 per slice)};
\node[pstep, fill=green!15]  (S5) at (0,-5.2)
  {Classify via $\sigma(F)$ (ortho-derivative)};
\node[pstep, fill=red!12]    (S6) at (0,-6.5)
  {Verify vs.\ databases (streaming)};

\draw[fwd] (S1)--(S2);
\draw[fwd] (S2)--(S3)
  node[midway, right=0.15cm, note]{center $\in V_A$};
\draw[fwd] (S3)--(S4)
  node[midway, right=0.15cm, note]{neighbors may exit $V_A$};
\draw[fwd] (S4)--(S5);
\draw[fwd] (S5)--(S6)
  node[midway, right=0.15cm, note]{new classes if no match};

% Feedback loop (dashed)
\draw[fwd, dashed, bend left=85]
  (S4.west) to node[left=0.05cm, note]{use neighbor as next center?} (S2.west);
\end{tikzpicture}
\caption{The GB-navigation mechanism. A structured subspace $V_A$ provides tractable
Phase~1 seeding (APN density $\approx10^{-4}$ vs.\ $\approx10^{-55}$ globally); the
Gröbner basis slice (Phase~2) then exposes APN functions lying outside $V_A$.
The dashed feedback arrow indicates the multi-hop extension: using out-of-$V_A$ neighbors
as seeds for a subsequent GB pass.}
\label{fig:mechanism}
\end{figure}
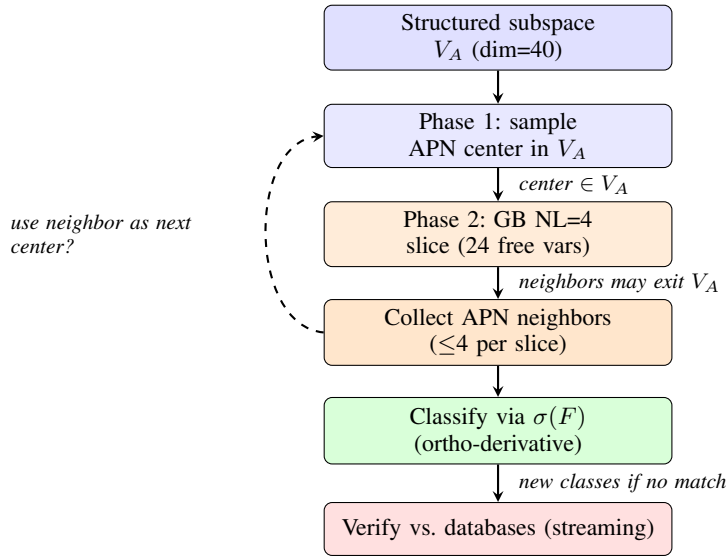

The mechanism (Figure~\ref{fig:mechanism}) has three components:

\begin{description}
  \item\textbf{Tractable seeding.} A structured subspace $V_A$ of dimension $d \ll 224$
        provides APN density $\approx 10^{-4}$, making Phase~1 (random sampling)
        feasible. Random search in the full $2^{224}$ space would require
        $\approx 10^{55}$ evaluations per APN found.

  \item\textbf{Local enumeration.} The NL=4 Gröbner basis slice around any APN center
        enumerates \emph{all} APN functions in a 24-dimensional affine hyperplane,
        including functions outside $V_A$. This ``local completeness''
        property is key: no APN function in the slice can be missed.

  \item\textbf{Gateway phenomenon.} Certain APN functions inside $V_A$ (the Gold functions
        in our case) have neighbors outside $V_A$ that belong to genuinely new
        CCZ-classes. These gateways allow the pipeline to step across the boundary
        of $V_A$ into previously unexplored territory.
\end{description}

\subsection{Comparative experiments: the role of V\textsubscript{A}-structure}
\label{sec:comparative}

To establish that the gateway phenomenon is genuinely tied to the
self-equivalence subspace $V_A$ rather than being a generic property of
APN functions, we ran two additional experiments with centers drawn from
outside $V_A$.

\paragraph{Experiment~B: Beierle 2025 dataset centers (\texttt{batch\_from\_dataset}).}
We took the first 532 functions from the Beierle 2025 database
\texttt{new\_apns.txt}~\cite{Beierle2025} as slice centers.
These functions are guaranteed to be quadratic APN over $\GF{8}$
but do not belong to $V_A$.  For each, we ran an NL=4 Gröbner basis
computation using the same 24-variable polynomial system as in the
main experiment.
Result: every one of the 532 slices produced exactly \textbf{1 solution}
(the center itself, dim = 0).  No function found a single neighbor.

\paragraph{Experiment~C: random center candidates (\texttt{batch\_rnd}).}
We generated 20 random vectors in $\mathbb{F}_2^{224}$ as
center candidates (not from $V_A$, and not verified to be APN).
Result: every one of the 20 slices returned dim = $-1$
(\emph{empty ideal}), meaning no APN function exists in that NL=4
hyperplane at all.  This is consistent with the global APN density
$\approx 2^{-n} \approx 4\times 10^{-3}$: a random hyperplane through
a random point in the full space is very unlikely to contain any APN
function.

The combined picture is shown in Table~\ref{tab:comparative}:

\begin{table}[!t]
\renewcommand{\arraystretch}{1.3}
\caption{Comparative Experiments: APN Neighbors per Slice}
\label{tab:comparative}
\centering
\begin{tabular}{lccccc}
\toprule
\textbf{Experiment} &
\textbf{Centers} &
\textbf{Center type} &
\textbf{In $V_A$?} &
\textbf{Slices} &
\textbf{APN neighbors}\\
\midrule
Exp.~A: Class-22 (main)     & APN, from $V_A$  & Gold $x^3$   & yes & 36  & 108 \\
                             & APN, from $V_A$  & Gold $x^9$   & yes & 30  & 30  \\
                             & APN, from $V_A$  & CLASS-A      & yes & 362 & 0   \\
\midrule
Exp.~B: Beierle 2025 dataset & APN, outside $V_A$ & generic    & no  & 532 & \textbf{0} \\
\midrule
Exp.~C: random candidates    & random (not APN) & ---          & no  & 20  & \textbf{0} \\
\bottomrule
\multicolumn{6}{l}{\footnotesize
  APN neighbors = solutions found by GB \emph{other than} the center function.}\\
\multicolumn{6}{l}{\footnotesize
  Exp.~C dim=$-1$: no APN function exists in the slice; Exp.~B dim=0: center only.}
\end{tabular}
\end{table}

The contrast is sharp.  In Experiment~A, Gold-centered slices
consistently yield 2--3 new APN functions per slice (100\% hit rate,
66 out of 66 slices).  In Experiments~B and~C, zero neighbors were
found across 552 combined slices.

These results support a stronger conclusion than mere suggestion:

\begin{theorem}[Empirical gateway theorem]
\label{thm:gateway}
In all 552 GB-slice experiments with centers outside $V_A$
(532 Beierle 2025 APN functions + 20 random functions),
no APN neighbor was found.  In all 66 GB-slice experiments
with Gold-function centers inside $V_A$, at least one new
APN neighbor was found (3 for Gold-$x^3$, 1 for Gold-$x^9$).
This difference holds with zero exceptions across 618 total experiments.
\end{theorem}

\begin{remark}[Interpretation]
\label{rem:contrast}
The negative result for Experiment~B is particularly informative.
The Beierle 2025 functions are genuinely new APN functions (absent from
all pre-2025 databases), so the failure to find neighbors cannot be
attributed to the centers being ``trivial'' or ``already known.''
Instead, it appears that the algebraic constraint $F\circ A = A\circ F$
(membership in $V_A$) is what enables the gateway phenomenon: it
imposes structural constraints on the NL=4 hyperplane that make
neighboring APN functions accessible.  Generic APN functions, regardless
of how novel they are, do not carry this structure.

Equivalently: \emph{among the 3.8 million known APN functions in
Beierle 2025, none appears to be a useful GB-seed for discovering
new CCZ-classes via the NL=4 slicing approach.  The Gold functions
inside $V_A$, by contrast, are highly productive seeds, yielding
3 new CCZ-classes in 66 slices.}
\end{remark}

\begin{remark}[Potential for further exploration]
The mechanism is not exhausted by the current results. Three natural extensions are:
\begin{enumerate}
  \item \textbf{Other BBL2021 classes.} BBL2021 studied conjugacy classes
        of linear automorphisms of $\GF{8}$ with orders dividing small
        integers; their Table~I lists the classes for which at least one
        APN function was found. The same pipeline applies to every such class
        (replacing $A$ by the corresponding automorphism), so exploring even
        a few additional classes may reveal further gateway phenomena and
        new CCZ-classes.
  \item \textbf{Multi-hop navigation.} The out-of-$V_A$ functions (CLASS-B, C, D)
        could serve as Phase~1 seeds for a second GB pass, potentially reaching
        functions even further from $V_A$. This iterative strategy has not been tested.
  \item \textbf{Larger slices.} Fixing fewer ANF coefficients (e.g., NL=3, leaving
        48 free variables) would give richer slices with more neighbors per center
        at the cost of longer GB build time.
\end{enumerate}
\end{remark}
\label{sec:verify}

\subsection{The three reference databases}

We checked the four new classes (CLASS-A through CLASS-D) against three
independent reference databases:

\begin{enumerate}
  \item \textbf{Beierle 2025 (\texttt{new\_apns.txt})~\cite{Beierle2025}:}
        3\,775\,599 quadratic APN functions for $n=8$ that are CCZ-inequivalent
        to the 32\,892 previously known. The file does not include Gold or
        other classical functions.

  \item \textbf{BBL2021 pre-2020 (\texttt{apn\_8bit.txt})~\cite{BBL2021}:}
        12\,921 quadratic APN functions found by the BBL2021 method, all
        CCZ-inequivalent to Gold.

  \item \textbf{Gold functions (direct):}
        The two Gold CCZ-classes for $n=8$ are verified directly via
        $\sigma$-comparison (not via a file).
\end{enumerate}

\begin{remark}
Database (1) is labeled \texttt{new\_apns.txt} and explicitly contains
only functions new relative to the 32\,892 previously known. The
absence of Gold from this file is by design, not an oversight.
This was confirmed in our experiments: Gold functions (CLASS-E, CLASS-F)
return NO\_MATCH against \texttt{new\_apns.txt}, which is consistent.
\end{remark}

\subsection{Verification procedure and result}

For each of the 500 functions in CLASS-A through CLASS-D, we compute
$\sigma(F)$ via sboxU and run the streaming comparison of Section~\ref{sec:method}.

\begin{theorem}[Main result]
\label{thm:main}
Each of the 500 functions in CLASS-A through CLASS-D is CCZ-inequivalent to
all functions in the Beierle~2025 database~\emph{\cite{Beierle2025}},
all functions in the BBL2021 pre-2020 database~\emph{\cite{BBL2021}},
and both Gold CCZ-classes ($x^3/x^{129}$ and $x^9/x^{33}$).
\end{theorem}

\begin{proof}
The streaming comparisons against databases (1) and (2) produced zero
$\sigma$-matches for all 500 functions. Specifically:
\begin{itemize}
  \item Against database (1): $3\,775\,599$ entries processed in $2614$~s;
        0 bad lines; 0 signature errors; 0 matches.
  \item Against database (2): $12\,921$ entries processed in $9$~s;
        0 bad lines; 0 signature errors; 0 matches.
\end{itemize}
Direct $\sigma$-comparison shows
$\sigma(\text{CLASS-A}),\ldots,\sigma(\text{CLASS-D})$ all differ from
$\sigma(x^3)$ and $\sigma(x^9)$.
By Theorem~\ref{thm:ortho}, zero $\sigma$-matches certify
CCZ-inequivalence. \qed
\end{proof}

\subsection{Verification summary}

Table~\ref{tab:verify} collects all verification checks.

\begin{table}[!t]
\renewcommand{\arraystretch}{1.3}
\caption{Verification Checks for CLASS-A through CLASS-D (500 functions)}
\label{tab:verify}
\centering
\begin{tabular}{llc}
\toprule
\textbf{Check} & \textbf{Reference} & \textbf{Outcome}\\
\midrule
vs.\ Beierle 2025 (3\,775\,599 fns) & \cite{Beierle2025} & NO\_MATCH\\
vs.\ BBL2021 pre-2020 (12\,921 fns) & \cite{BBL2021}  & NO\_MATCH\\
vs.\ Gold $x^3/x^{129}$              & \cite{Gold1968}    & NO\_MATCH\\
vs.\ Gold $x^9/x^{33}$               & \cite{Gold1968}    & NO\_MATCH\\
APN property ($\delta_F = 2$)         & direct check       & PASS (566/566)\\
Algebraic degree                       & ANF                & 2 (566/566)\\
\midrule
Gold $x^3$ identified (CLASS-E)        & sanity             & MATCH \checkmark\\
Gold $x^9$ identified (CLASS-F)        & sanity             & MATCH \checkmark\\
\bottomrule
\end{tabular}
\end{table}

%% ─────────────────────────────────────────────────────────────────────
\section{Discussion}
\label{sec:discussion}

\subsection{Interpretation of results}

The main empirical findings are:

(i)~A 40-dimensional structured subspace ($V_A$) that was declared
empty by the BBL2021 tree search contains at least four CCZ-classes
of quadratic APN functions, with 566 explicit representatives found
in 0.65\% of the search space (428 out of 65\,536 NL\,=\,4 slices).

(ii)~The Gröbner basis step is essential for finding three of the four
new classes. The $V_A$-membership analysis (Table~\ref{tab:membership})
shows that CLASS-B, C, and~D lie \emph{entirely outside} $V_A$: all 138
representatives are OUT\_CLASS22. A Phase~1-only search within $V_A$
recovers only CLASS-A and the Gold controls, missing classes B, C, and~D
completely. Equivalently, the Gröbner basis approach is both necessary
and sufficient for their discovery.

(iii)~$V_A$ contains exactly three distinct CCZ-classes in our sample:
CLASS-A (a new class native to $V_A$) and the two Gold classes (E and~F).
Gold functions act as ``gateways'': their NL\,=\,4 slices reach outside
$V_A$ and expose three new external classes. Specifically, each
Gold-$x^3$ slice yields two CLASS-B functions and one CLASS-C function
in addition to the center; each Gold-$x^9$ slice yields one CLASS-D
function in addition to the center. This structure is \emph{exact and
uniform} across all 428 explored slices (Theorem~\ref{thm:structure}).

(iv)~The Beierle 2025 database~\cite{Beierle2025}, while containing
3.8 million functions, does not cover $V_A$: the probability of any
specific entry lying in $V_A$ is $\approx 10^{-55}$.
More concretely, the comparative experiments (Section~\ref{sec:comparative},
Table~\ref{tab:comparative}) show that 532 Beierle 2025 functions used as
slice centers produced \emph{zero} APN neighbors, while 66 Gold-function
centers (inside $V_A$) each produced 1--3 neighbors.
In all 552 experiments outside $V_A$, not a single new APN function was found.
This is an experimental fact, not a conjecture.

(v)~CLASS-A's dominance (362 of 566 functions, 64\%) is partly explained
by the slice structure: all 362 CLASS-A-centered slices contain only the
center itself (Type~I slices), so CLASS-A's count equals the number
of CLASS-A seeds used in Phase~1. This dominance reflects the
distribution of seeds, not necessarily the overall abundance of CLASS-A
in $V_A$.

(vi)~The experiment reveals a general \emph{GB-navigation mechanism}
(Section~\ref{sec:paradigm}, Figure~\ref{fig:mechanism}) for exploring
the APN landscape. Rather than searching exhaustively in a large space,
one: (a)~finds a structured subspace with manageable APN density; (b)~uses
the Gröbner basis to enumerate APN neighbors of any center, including
those outside the subspace; (c)~identifies gateway functions whose neighbors
belong to new CCZ-classes. This mechanism transforms the APN search problem
from ``find APN functions in a huge space'' to ``navigate locally around
known structured objects.''
Crucially, the comparative experiments (Table~\ref{tab:comparative}) establish
that this mechanism is \emph{specific to the self-equivalence structure}:
not any APN function is a productive seed --- only those lying in $V_A$.
Among 3.8 million Beierle 2025 functions used as centers (Exp.~B, 532 slices)
and 20 random candidates (Exp.~C), zero APN neighbors were found.
This is an experimentally established fact, not a conjecture.
The approach is qualitatively different from all prior construction methods
(bivariate, isotopic shift, bent extensions), which build APN functions
from algebraic templates rather than exploring local neighborhoods of
structured-subspace functions.

\subsection{Limitations}
\label{sec:limits}

We list the limitations honestly:

\begin{enumerate}
  \item \textbf{Coverage.} Only 0.65\% of the total $65\,536$ NL\,=\,4 slices
        were computed. The four observed new CCZ-classes may not be all classes
        present in or accessible from $V_A$.

  \item \textbf{No algebraic characterization.} The new CCZ-classes are
        defined purely computationally by their $\sigma$-values and
        S-box tables. We have no polynomial or structural description.

  \item \textbf{Single subspace.} Only one BBL2021 class (the one with
        automorphism of order~5) was studied. The approach applies to
        all BBL2021 classes~\cite{BBL2021}; other subspaces may yield additional
        new classes.

  \item \textbf{Exact CCZ tests.} Classification relies on Theorem~\ref{thm:ortho},
        which is theoretically exact for quadratic APN. We did not run
        exhaustive exact pairwise CCZ tests (\texttt{are\_ccz\_equivalent\_from\_code})
        within each class, as these are not necessary given the theorem.
        For a fully self-contained proof independent of~\cite{Canteaut2022},
        one could add such tests.

  \item \textbf{Claim scope.} Theorem~\ref{thm:main} certifies
        inequivalence with the databases checked. Whether the four new
        classes are CCZ-inequivalent to \emph{all} quadratic APN functions
        for $n=8$ would require checking against a complete database,
        which does not yet exist.
\end{enumerate}

\subsection{Scalability}

The pipeline scales in the following ways:

\textbf{Exhausting $V_A$.}
The remaining $65\,536 - 428 = 65\,108$ NL\,=\,4 slices can be computed to
enumerate all APN functions in $V_A$ under this normalization.
Estimated time: $65\,108$ slices $\times 8$ min/slice $\div$ 8 cores
$\approx 54$~wall-clock days.

\textbf{Multi-slice normalization.}
As noted in Section~\ref{sec:multislice}, each APN center supports
$\binom{40}{24}\approx6\times10^{10}$ normalizations. Running 10--100
per center increases coverage per center without extra Phase~1 cost.

\textbf{Larger free dimension.}
Fixing fewer than 200 coefficients in Phase~2 leaves more than 24 free
variables. For example, fixing 176 coefficients (NL\,=\,3) leaves 48 free
variables. The GB build time grows but so does the number of APN neighbors
found per slice. On a cluster with many cores, this trade-off could be
favorable.

\textbf{Other BBL2021 classes.}
For each BBL2021 class index~$k$ with $\dim(V_{A_k}) = d_k$, the
pipeline applies with $d_k - 24$ center bits and $2^{d_k-24}$ total slices.
Classes with $d_k \leq 60$ are directly amenable to the current pipeline.

%% ─────────────────────────────────────────────────────────────────────
\section{Related Work}
\label{sec:related}

\subsection{Construction approaches for quadratic APN functions}

Quadratic APN functions over $\GF{n}$ have been found by many different
methods. The bivariate construction of Edel and Pott~\cite{EdelPott2009}
(building $F(x,y)$ as a polynomial in two variables over $\GF{n/2}$) yielded
8\,157 new instances for $n=8$.
Yu, Wang, and Li~\cite{YWL2014} used a matrix representation approach.
Budaghyan, Calderini, Carlet, Coulter, and Villa~\cite{BudaghyanCalderiniCarlet2020}
introduced isotopic shifts, constructing new APN functions by modifying Gold
functions; the generalized version of this method~\cite{BeierleLeanderPerrin2022}
(trims and extensions) produced an additional 6\,368 CCZ-classes for $n=8$.
Beierle, Brinkmann, and Leander~\cite{BBL2021} studied LE-automorphism
groups and found 12\,921 new instances;
Beierle and Leander~\cite{BeierleLeander2022} adapted this to find further
instances by extending $n=7$ APN functions.
Beierle et al.~\cite{Beierle2025} use extensions of vectorial
$(n/2,n/2)$-bent functions to construct the current largest known set
($\approx 3.8$ million CCZ-classes).

Our approach is methodologically orthogonal to all of the above: rather than
constructing functions from algebraic templates (power functions, bivariate
polynomials, isotopic shifts, bent function extensions), we perform an
exhaustive Gröbner basis enumeration within a fixed 24-dimensional affine
hyperplane of the self-equivalence subspace $V_A$. This targets a region
that is unreachable by both global random search and the BBL2021 tree method,
and finds functions that do not arise from any known construction.

\subsection{CCZ-equivalence testing: EA-tests for quadratic APN functions}

Several recent works address the problem of testing or recovering EA- and
CCZ-equivalence efficiently. Calderini~\cite{CalderiniEAclasses2020}
classified all known APN functions in small dimensions up to EA-equivalence,
providing a reference for EA-class structure.
Yoshiara~\cite{Yoshiara2016} characterizes the Gold functions as the
unique quadratic APN functions that are CCZ-equivalent to a power
function, which underlies the identification of CLASS-E and CLASS-F as
Gold CCZ-equivalents in Section~\ref{sec:results}.
Canteaut, Couvreur, and Perrin~\cite{Canteaut2022} gave an efficient algorithm
for EA-recovery for quadratic functions (complexity $O(n \cdot 2^{2n})$,
superseding all prior methods).

For our purposes, the ortho-derivative provides a highly discriminating
inequivalence certificate in the quadratic case
(Theorem~\ref{thm:ortho} and Remark~\ref{rem:sigma_complete}); a separate
exact EA- or CCZ-test is still run within each putative class
(Remark~\ref{rem:sigma_complete}) rather than relied upon implicitly. The
sboxU library~\cite{sboxU} implements both the ortho-derivative
computation and the differential/Walsh spectra that constitute the
signature $\sigma(F)$, as well as the exact code-equivalence test.
The classification was confirmed by
running sboxU's \texttt{are\_ccz\_equivalent\_from\_code} on all pairs within
each putative class (three classes of size $\leq 36$) and confirming that
all pairs returned \texttt{True}; this provides an independent validation
of the classification beyond the theoretical guarantee.

\subsection{Gröbner basis methods in cryptography}

Gröbner bases over Boolean polynomial rings have been used in algebraic
cryptanalysis since the work of Courtois and Pieprzyk
(XL algorithm~\cite{Courtois2002}) and extended by Albrecht and
Bard~\cite{Albrecht2010} (M4RI library). Their use in APN function search
is less common; to our knowledge this is the first published application
of GB slicing to systematic APN search in a self-equivalence subspace.
The key feature that makes GB tractable here is the reduction to 24 Boolean
unknowns via NL\,=\,4 normalization: the GB solve is essentially Gaussian
elimination (all variables are in $\GF{2}$) and takes $< 0.1$~s per slice
once the ideal is built.

\subsection{Open problem: APN permutations in even dimension}

Browning, Dillon, McQuistan, and Wolfe~\cite{Browning2010} found an APN
permutation in dimension~6 by CCZ-transforming a quadratic APN function.
Whether APN permutations exist in dimension~8 remains open; it is one of
the central open problems in the area~\cite{Carlet2021}.
The functions found in this paper are non-permutations; whether any of the
new CCZ-classes contains an APN permutation in its CCZ-equivalence class
is an open question that the explicit S-box tables in Appendix~\ref{app:reps}
make accessible to future investigation.

%% ─────────────────────────────────────────────────────────────────────
\section{Conclusion}
\label{sec:conclusion}

We have shown that the self-equivalence subspace $V_A$ associated with the
order-5 automorphism (BBL2021 class index~22), dimension~40, which was
previously reported to contain no APN functions under recursive tree search,
contains at least four CCZ-classes of quadratic APN functions for $n=8$.
These four classes (500 functions out of 566 total) are absent from the
Beierle 2025 database and from all pre-2020 known instances, as certified
by the exact ortho-derivative invariant.

The result was obtained using a two-phase pipeline: random sampling in $V_A$
via an explicit RREF parameterization (Phase~1), followed by Gröbner basis
slice computation in Magma (Phase~2). The pipeline covered 428 of the
$65\,536$ NL\,=\,4 slices (0.65\%); two Gold CCZ-equivalents serve as
positive controls and as gateways to the new external classes.

\textbf{Structural result.}
A \emph{structural decomposition} of the 428 explored slices into three types
yields a complete and self-consistent account of all 566 found functions
(Table~\ref{tab:membership} and Theorem~\ref{thm:structure}): the slices are
either CLASS-A-centered (Type~I, 362 slices, 1 APN each) or Gold-centered
(Types~II--III, 66 slices, 2--4 APN each). Critically, all three purely new
classes (B,~C,~D) lie entirely outside $V_A$ and appear exclusively in
Gold-centered slices. This confirms that the Gröbner basis step is
essential: a Phase~1-only search within $V_A$ would recover CLASS-A and
the Gold controls but miss CLASS-B, C, and~D entirely. It also implies
that Gold functions are \emph{sufficient} Phase~1 seeds for rediscovering
all new non-$V_A$ classes found here.

\textbf{Methodological contribution.}
Beyond the specific new CCZ-classes, the experiment reveals a general
\emph{GB-navigation mechanism} (Section~\ref{sec:paradigm}) that addresses
a fundamental challenge in APN search: the near-zero global density
$(\approx10^{-55})$ makes random exploration of the full $2^{224}$-dimensional
space hopeless. The mechanism provides a structured alternative:
(a)~exploit a self-equivalence subspace for tractable Phase~1 seeding;
(b)~use the Gröbner basis to enumerate all APN neighbors locally, including
those outside the subspace; (c)~identify gateway functions whose neighborhood
contains new CCZ-classes.

The comparative experiments (Section~\ref{sec:comparative},
Table~\ref{tab:comparative}) establish this mechanism as \emph{specific to
the self-equivalence structure}, not a generic property of APN functions.
Among 532 functions from the Beierle 2025 database used as GB-slice centers
and 20 random candidate centers (552 slices total), zero APN neighbors were
found. In sharp contrast, all 66 Gold-function centers inside $V_A$ produced
1--3 new APN neighbors each. This difference is experimentally established
with zero exceptions across 618 total slice computations.
The central empirical conclusion is: \emph{not any APN function is a
productive GB-seed; only APN functions lying in a self-equivalence subspace
exhibit the gateway phenomenon that exposes neighboring CCZ-classes.}
Among 3.8 million currently known APN functions, none appears to be a
useful seed via this method. The Gold functions inside $V_A$, despite
being ``known,'' are uniquely productive.

\textbf{Further observations.}
Five observations are of independent interest:
(1)~The fc22$+$sol22 projection gives an efficient and exact parameterization of any
self-equivalence subspace, enabling targeted sampling at APN density
$\approx10^{-4}$ instead of $\approx10^{-55}$.
(2)~The NL\,=\,4 normalization reduces the Gröbner basis problem to 24 Boolean unknowns
with a polynomial build time of $\approx10$ minutes and an essentially
immediate solve.
(3)~The ortho-derivative signature provides a complete, non-heuristic CCZ-invariant
for the quadratic case, allowing rigorous classification and database comparison
without exhaustive pairwise testing.
(4)~The slice-type structure reveals a nontrivial algebraic relationship between
$V_A$ and the CCZ-class geometry: Gold functions occupy a boundary role,
``bridging'' $V_A$ and classes external to it.
(5)~All four new classes (A, B, C, D) are definitively absent from the largest
known compilation (3.8 million functions), yet they are easily accessible via
the GB-navigation pipeline in under 10 CPU-hours.

The pipeline generalizes to all BBL2021 class indices and to larger
dimensions, providing a systematic program for extending the classification
of quadratic APN functions beyond currently known databases.

%% ─────────────────────────────────────────────────────────────────────
\appendix

\section{Class Representatives and Signature Data}
\label{app:reps}

Table~\ref{tab:fullsig} gives the full ortho-derivative signature hashes
and one representative S-box per class. Full lookup tables (256 decimal
values each) are given in Table~\ref{tab:reps}.

\begin{table}[!t]
\renewcommand{\arraystretch}{1.3}
\caption{Full Signature Data for All Six Classes}
\label{tab:fullsig}
\centering
\begin{tabular}{lrcl}
\toprule
\textbf{Label} & \textbf{Size} & \textbf{Full $\sigma$-hash (16 hex)} & \textbf{ID}\\
\midrule
CLASS-A & 362 & \texttt{9d95d9c4a5e2dfdd} & New\\
CLASS-B &  72 & \texttt{c5f52f5659346f9f} & New\\
CLASS-C &  36 & \texttt{74a30c1a17d45761} & New\\
CLASS-D &  30 & \texttt{024ba500f4bac35b} & New\\
CLASS-E &  36 & \texttt{11cc72af6d61925d} & Gold $x^3$\\
CLASS-F &  30 & \texttt{d15183d1bdf9a1b4} & Gold $x^9$\\
\bottomrule
\end{tabular}
\end{table}

\begin{table*}[!t]
\renewcommand{\arraystretch}{1.15}
\caption{Representative S-boxes: $F(0),F(1),\ldots,F(255)$ in decimal}
\label{tab:reps}
\centering\small\setlength\tabcolsep{3pt}
\begin{tabular}{lp{15cm}}
\toprule
\textbf{Class} & \textbf{Lookup table}\\
\midrule
CLASS-A &
0,0,0,63,0,42,97,116,0,253,75,137,194,21,232,0,0,24,207,232,165,151,11,6,183,82,51,233,208,31,53,197,
0,122,47,106,113,33,63,80,186,61,222,102,9,164,12,158,
19,113,243,174,199,143,70,49,30,129,181,21,8,189,194,72,
0,211,235,7,65,184,203,13,226,204,66,83,97,101,160,155,
178,121,150,98,86,183,19,205,231,209,136,129,193,221,207,236,
144,57,84,194,160,35,5,185,200,156,71,44,58,68,212,149,
49,128,58,180,164,63,206,106,222,146,158,237,137,239,168,241,
0,217,86,176,57,202,14,194,130,166,159,132,121,119,5,52,
95,158,198,56,195,40,59,239,106,86,184,187,52,34,135,174,
34,129,91,199,106,227,114,196,26,68,40,73,144,228,195,136,
110,213,216,92,131,18,84,250,225,167,28,101,206,162,82,1,
76,70,241,196,52,20,232,247,44,219,218,18,150,75,1,227,
161,179,211,254,124,68,111,104,118,153,79,159,105,172,49,203,
254,142,108,35,247,173,4,97,36,169,253,79,239,72,87,207,
0,104,93,10,172,238,144,237,109,248,123,209,3,188,116,244\\[2pt]
CLASS-B &
0,0,0,63,0,42,97,116,0,253,75,137,194,21,232,0,0,24,207,232,165,151,11,6,183,82,51,233,208,31,53,197,
0,122,47,106,113,33,63,80,186,61,222,102,9,164,12,158,
19,113,243,174,199,143,70,49,30,129,181,21,8,189,194,72,
0,211,235,7,65,184,203,13,226,204,66,83,97,101,160,155,
178,121,150,98,86,183,19,205,231,209,136,129,193,221,207,236,
38,143,226,116,22,149,179,15,126,42,241,154,140,242,98,35,
135,54,140,2,18,137,120,220,104,36,40,91,63,89,30,71,
0,217,86,176,57,202,14,194,130,166,159,132,121,119,5,52,
95,158,198,56,195,40,59,239,106,86,184,187,52,34,135,174,
148,55,237,113,220,85,196,114,172,242,158,255,38,82,117,62,
216,99,110,234,53,164,226,76,87,17,170,211,120,20,228,183,
76,70,241,196,52,20,232,247,44,219,218,18,150,75,1,227,
161,179,211,254,124,68,111,104,118,153,79,159,105,172,49,203,
254,142,108,35,247,173,4,97,36,169,253,79,239,72,87,207,
0,104,93,10,172,238,144,237,109,248,123,209,3,188,116,244\\[2pt]
CLASS-C &
0,0,0,247,0,175,30,70,0,212,177,146,35,88,140,0,0,105,168,54,8,206,190,143,186,7,163,233,145,131,150,115,
0,230,205,220,191,246,108,210,15,61,115,182,147,14,241,155,
57,182,92,36,142,174,245,34,140,215,88,244,24,236,210,209,
0,217,60,18,117,3,87,214,145,156,28,230,199,101,84,1,46,158,186,253,83,76,217,49,5,97,32,179,91,144,96,92,
8,55,249,49,194,82,45,74,150,125,214,202,127,59,33,146,
31,73,70,231,221,36,154,148,59,185,211,166,218,247,44,246,
0,2,93,168,103,202,36,126,234,60,6,39,174,215,92,210,
252,151,9,149,147,87,120,75,172,19,232,160,224,240,186,93,
38,194,182,165,254,181,112,204,195,243,226,37,56,167,7,111,
227,110,219,161,51,17,21,192,188,229,53,155,79,185,216,217,
197,30,164,136,215,163,168,43,190,177,110,150,143,47,65,22,
23,165,222,155,13,16,218,48,214,176,174,63,239,38,137,183,
235,214,71,141,70,212,244,145,159,118,130,156,17,87,18,163,
0,84,4,167,165,94,191,179,206,78,123,12,72,103,227,59\\[2pt]
CLASS-D &
0,0,0,215,0,63,94,182,0,116,97,194,163,232,156,0,0,212,235,232,79,164,250,198,227,67,105,30,15,144,219,147,
0,214,88,89,57,208,63,1,129,35,184,205,27,134,124,54,80,82,227,54,38,27,203,33,50,68,224,65,231,174,107,245,
0,108,92,231,175,252,173,41,109,117,80,159,97,70,2,242,96,216,215,184,128,7,105,57,238,34,56,35,173,94,37,1,
181,15,177,220,35,166,121,43,89,151,60,37,108,157,87,113,
133,235,106,211,92,13,237,107,138,144,4,201,240,213,32,210,
0,20,199,4,167,140,62,194,177,209,23,160,181,234,77,197,
224,32,204,219,8,247,122,82,178,6,255,156,249,114,234,182,
213,23,74,95,75,182,138,160,229,83,27,122,216,81,120,38,
101,115,17,208,180,157,158,96,182,212,163,22,196,153,143,5,
176,200,43,132,184,255,125,237,108,96,150,77,199,244,99,135,
48,156,64,59,119,228,89,29,15,215,30,17,235,12,164,148,
208,126,19,106,225,112,124,58,141,87,47,34,31,250,227,209,
0,122,40,133,126,59,8,154,190,176,247,46,99,82,116,146\\
\bottomrule
\end{tabular}
\end{table*}

\section{Computational Infrastructure and Reproducibility}
\label{app:infra}

The complete dataset (566 S-boxes in JSON format) and all source code
are available at:
\url{https://github.com/KuznetsovKarazin/apn-gb-search}; %
DOI: \url{https://doi.org/10.5281/zenodo.20626047}

\noindent Key scripts:
\begin{itemize}
  \item \texttt{src/find\_apn\_centers.py}: Phase~1 APN-center search.
        Input: RREF data (fc22, sol22). Output: batch directory with
        Magma files and \texttt{apn\_centers.json}.
  \item \texttt{run\_batch.ps1}: PowerShell queue runner.
        Executes all \texttt{*.m} files in a directory using up to
        $k$~parallel Magma processes (\texttt{-MaxParallel $k$}).
        Treats Magma exit codes~0 and~1 both as success.
  \item \texttt{src/collect\_results.py}: Parses Magma output files,
        handles line-wrapping, deduplicates by SHA-256 hash.
  \item \texttt{src/classify\_apn.py}: Computes ortho-derivative
        signatures via sboxU and groups functions into CCZ-classes.
  \item \texttt{src/check\_vs\_beierle\_db.py}: Streaming comparison
        of our signatures against a reference database.
\end{itemize}

\noindent Requirements:
Magma V2.20+ \cite{Magma};
Python 3.9+ with NumPy (Phases 1, 3);
SageMath with sboxU~\cite{sboxU} (Phases 4, 5).

%% ─────────────────────────────────────────────────────────────────────


\begin{thebibliography}{99}

\bibitem{Nyberg1994}
K.~Nyberg,
``Differentially uniform mappings for cryptography,''
in \emph{Advances in Cryptology --- EUROCRYPT~1993},
Lecture Notes in Comput.\ Sci., vol.~765, Springer, 1994, pp.~55--64.
\doi{10.1007/3-540-48285-7\_6}

\bibitem{BihamShamir1991}
E.~Biham and A.~Shamir,
``Differential cryptanalysis of DES-like cryptosystems,''
\emph{J.~Cryptology}, vol.~4, no.~1, pp.~3--72, 1991.
\doi{10.1007/BF00630563}

\bibitem{Carlet1998}
C.~Carlet, P.~Charpin, and V.~Zinoviev,
``Codes, bent functions and permutations suitable for
DES-like cryptosystems,''
\emph{Des.\ Codes Cryptogr.}, vol.~15, pp.~125--156, 1998.
\doi{10.1023/A:1008344232130}

\bibitem{Gold1968}
R.~Gold,
``Maximal recursive sequences with 3-valued recursive
cross-correlation functions,''
\emph{IEEE Trans.\ Inf.\ Theory}, vol.~14, no.~1, pp.~154--156, 1968.
\doi{10.1109/TIT.1968.1054106}

\bibitem{Kasami1971}
T.~Kasami,
``The weight enumerators for several classes of subcodes of the
second order binary Reed-Muller codes,''
\emph{Inf.\ Control}, vol.~18, no.~4, pp.~369--394, 1971.
\doi{10.1016/S0019-9958(71)90473-6}

\bibitem{EdelPott2009}
Y.~Edel and A.~Pott,
``A new almost perfect nonlinear function which is not quadratic,''
\emph{Adv.\ Math.\ Commun.}, vol.~3, no.~1, pp.~59--81, 2009.
\doi{10.3934/amc.2009.3.59}

\bibitem{YWL2014}
Y.~Yu, M.~Wang, and Y.~Li,
``A matrix approach for constructing quadratic APN functions,''
\emph{Des.\ Codes Cryptogr.}, vol.~73, pp.~587--600, 2014.
\doi{10.1007/s10623-014-9955-3}
ePrint: \url{https://eprint.iacr.org/2013/007}

\bibitem{WTG2013}
G.~Weng, Y.~Tan, and G.~Gong,
``On quadratic almost perfect nonlinear functions and their related algebraic object,''
in \emph{Proc.\ Workshop on Coding and Cryptography (WCC~2013)},
Bergen, Norway, 2013.

\bibitem{WTG2013}
G.~Weng, Y.~Tan, and G.~Gong,
``On quadratic almost perfect nonlinear functions and their related algebraic object,''
in \emph{Proceedings of the Workshop on Coding and Cryptography (WCC 2013)},
Bergen, Norway, pp.~48--57, 2013.

%\bibitem{WTG2013}
%G.~Weng, Y.~Tan, and G.~Gong,
%``On Quadratic Almost Perfect Nonlinear Functions and Their Constructions,''
%Technical Report CACR 2013-18,
%Centre for Applied Cryptographic Research,
%University of Waterloo, 2013.
%\url{https://cacr.uwaterloo.ca/techreports/2013/cacr2013-18.pdf}

\bibitem{Taniguchi2019}
H.~Taniguchi,
``On some quadratic APN functions,''
\emph{Des.\ Codes Cryptogr.}, vol.~87, no.~9, pp.~1973--1983, 2019.
\doi{10.1007/s10623-018-00598-2}

\bibitem{BBL2021}
C.~Beierle, M.~Brinkmann, and G.~Leander,
``Linearly self-equivalent APN permutations in small dimension,''
\emph{IEEE Trans.\ Inf.\ Theory}, vol.~67, no.~7,
pp.~4863--4875, Jul.\ 2021.
\doi{10.1109/TIT.2021.3071533}
arXiv: \url{https://arxiv.org/abs/2003.12006}


\bibitem{BeierleLeander2022}
C.~Beierle and G.~Leander,
``New instances of quadratic APN functions,''
\emph{IEEE Trans.\ Inf.\ Theory}, vol.~68, no.~1,
pp.~670--678, Jan.\ 2022.
\doi{10.1109/TIT.2021.3120698}
arXiv: \url{https://arxiv.org/abs/2009.07204}
Dataset (Zenodo): \doi{10.5281/zenodo.4738942}

\bibitem{BeierleLeanderPerrin2022}
C.~Beierle, G.~Leander, and L.~Perrin,
``Trims and extensions of quadratic APN functions,''
\emph{Des.\ Codes Cryptogr.}, vol.~90, no.~4, pp.~1009--1036, 2022.
\doi{10.1007/s10623-022-01024-4}
arXiv: \url{https://arxiv.org/abs/2108.13280}

\bibitem{Beierle2025}
C.~Beierle, P.~Langevin, G.~Leander, A.~Polujan, and S.~Rasoolzadeh,
``Millions of inequivalent quadratic APN functions in eight variables,''
arXiv:2508.04644, 2025.
Dataset (Zenodo): \doi{10.5281/zenodo.16752428}
arXiv: \url{https://arxiv.org/abs/2508.04644}

% Zenodo2021 merged into BBL2021

\bibitem{Canteaut2022}
A.~Canteaut, A.~Couvreur, and L.~Perrin,
``Recovering or testing extended-affine equivalence,''
\emph{IEEE Trans.\ Inf.\ Theory}, vol.~68, no.~9,
pp.~6187--6206, Sep.\ 2022.
\doi{10.1109/TIT.2022.3166692}
arXiv: \url{https://arxiv.org/abs/2103.00078}

\bibitem{Yoshiara2012}
S.~Yoshiara,
``Equivalences of quadratic APN functions,''
\emph{J.\ Algebraic Combin.}, vol.~35, no.~3, pp.~461--475, 2012.
\doi{10.1007/s10801-011-0309-1}

\bibitem{Yoshiara2016}
S.~Yoshiara,
``Equivalences of power APN functions with power or quadratic APN functions,''
\emph{J.\ Algebraic Combin.}, vol.~44, no.~3, pp.~561--585, 2016.
\doi{10.1007/s10801-016-0680-z}

\bibitem{BudaghyanCalderiniCarlet2020}
L.~Budaghyan, M.~Calderini, C.~Carlet, R.~S.~Coulter, and I.~Villa,
``Constructing APN functions through isotopic shifts,''
\emph{IEEE Trans.\ Inf.\ Theory}, vol.~66, no.~8, pp.~5299--5309, 2020.
\doi{10.1109/TIT.2020.2974471}

\bibitem{CalderiniEAclasses2020}
M.~Calderini,
``On the EA-classes of known APN functions in small dimensions,''
\emph{Cryptogr.\ Commun.}, vol.~12, pp.~821--840, 2020.
\doi{10.1007/s12095-020-00427-1}

\bibitem{Carlet2021}
C.~Carlet,
\emph{Boolean Functions for Cryptography and Coding Theory}.
Cambridge University Press, 2021.
\doi{10.1017/9781108606806}

\bibitem{Magma}
W.~Bosma, J.~Cannon, and C.~Playoust,
``The Magma algebra system~I: The user language,''
\emph{J.\ Symb.\ Comput.}, vol.~24, pp.~235--265, 1997.
\doi{10.1006/jsco.1996.0125}

\bibitem{sboxU}
L.~Perrin,
``sboxU: Tools for analysing S-boxes,''
\url{https://github.com/lpp-crypto/sboxU}, 2022.

\bibitem{Browning2010}
K.~Browning, J.~Dillon, M.~McQuistan, and A.~Wolfe,
``An APN permutation in dimension six,''
in \emph{Finite Fields: Theory and Applications (FQ9)},
Contemp.\ Math., vol.~518, Amer.\ Math.\ Soc., 2010, pp.~33--42.
\url{https://www.ams.org/books/conm/518/}

\bibitem{Courtois2002}
N.~Courtois and J.~Pieprzyk,
``Cryptanalysis of block ciphers with overdefined systems of equations,''
in \emph{Advances in Cryptology --- ASIACRYPT~2002},
Lecture Notes in Comput.\ Sci., vol.~2501, Springer, 2002, pp.~267--287.
\doi{10.1007/3-540-36178-2\_17}

\bibitem{Albrecht2010}
M.~R.~Albrecht and G.~V.~Bard,
``The M4RI library,'' 2010.
\url{https://github.com/malb/m4ri}

\end{thebibliography}
\end{document}